\documentclass[article,notitlepage,amsthm,amsmath,amssymb]{revtex4-2}  
\usepackage{color}
\usepackage[dvipsnames]{xcolor}
\usepackage{graphicx}
\usepackage{subcaption}
\usepackage{amsmath}
\usepackage{amssymb}
\usepackage{amsthm}
\usepackage[colorlinks = true,linkcolor = blue,urlcolor = blue,citecolor = blue,anchorcolor = blue]{hyperref}
\usepackage{pgfplots}
\usepackage{comment}
\usepackage{algpseudocode}
\usepackage{algorithm}
\usepackage{{booktabs}}

\usepackage{xy}
\xyoption{matrix}
\xyoption{frame}
\xyoption{arrow}
\xyoption{arc}

\usepackage{ifpdf}
\ifpdf
\else
\PackageWarningNoLine{Qcircuit}{Qcircuit is loading in Postscript mode.  The Xy-pic options ps and dvips will be loaded.  If you wish to use other Postscript drivers for Xy-pic, you must modify the code in Qcircuit.tex}
\xyoption{ps}
\xyoption{dvips}
\fi

\entrymodifiers={!C\entrybox}

\newcommand{\ket}[1]{{\left\vert{#1}\right\rangle}}

\newtheorem{lemma}{Lemma}

\renewcommand{\sec}[1]{\hyperref[sec:#1]{Section~\ref*{sec:#1}}}
\newcommand{\ssec}[1]{\hyperref[ssec:#1]{Subsection~\ref*{ssec:#1}}}
\newcommand{\fig}[1]{\hyperref[fig:#1]{Figure~\ref*{fig:#1}}}
\newcommand{\tab}[1]{\hyperref[tab:#1]{Table~\ref*{tab:#1}}}
\newcommand{\lem}[1]{\hyperref[lem:#1]{Lemma~\ref*{lem:#1}}}
\newcommand{\propos}[1]{\hyperref[propos:#1]{Proposition~\ref*{propos:#1}}}
\newcommand{\thm}[1]{\hyperref[thm:#1]{Theorem~\ref*{thm:#1}}}
\newcommand{\alg}[1]{\hyperref[alg:#1]{Algorithm~\ref*{alg:#1}}}

\newcommand{\Gate}[1]{\textsc{#1}}

\newcommand{\cnotgate}{\Gate{CNOT}}

\newcommand{\toffoligate}{\Gate{Toffoli}}
\newcommand{\czgate}{\Gate{CZ}}

\newcommand{\cczgate}{\Gate{CCZ}}
 
\newcommand{\swapgate}{\Gate{SWAP}} 
\newcommand{\hgate}{\Gate{H}}

\newcommand{\pgate}{\Gate{P}} 

\newcommand{\tgate}{\Gate{T}}

\pgfplotsset{compat=1.18}

\begin{document}

\title{Asymptotic yet practical optimization of quantum circuits implementing GF$(2^m)$ multiplication and division operations}

\author{Noureldin Yosri}
\affiliation{Google Quantum AI, 350 Main St, Venice, CA 90291, United States}
\author{Dmytro Gavinsky}
\affiliation{Institute of Mathematics of the Czech Academy of Sciences, \v Zitna 25, Praha 1, Czech Republic}
\author{Dmitri Maslov}
\affiliation{Google Quantum AI, 25 Massachusetts Ave NW, Washington, DC 20001, United States}

\begin{abstract}
We present optimized quantum circuits for GF$(2^m)$ multiplication and division operations, which are essential computing primitives in various quantum algorithms.  Our ancilla-free GF multiplication circuit has the gate count complexity of $O(m^{\log_2{3}})$, an improvement over the previous best bound of $O(m^2)$.  This was achieved by developing an efficient $O(m)$ circuit for multiplication by the constant polynomial $1+x^{\lceil{m/2}\rceil}$, a key component of Van Hoof's construction \cite{van2019space}.  This asymptotic reduction translates to a factor of 100+ improvement of the $\cnotgate$ gate counts in the implementation of the multiplication by the constant for parameters $m$ of practical importance.  For the GF division, we reduce gate count complexity from $O(m^2 \log(m))$ to $O(m^2 \log \log(m)/\log(m))$ by selecting irreducible polynomials that enable efficient implementation of both the constant polynomial multiplication and field squaring operations.  We demonstrate practical advantages for cryptographically relevant values of $m$, including reductions in both $\cnotgate$ and $\toffoligate$ gate counts.  Additionally, we explore the complexity of implementing square roots of linear reversible unitaries and demonstrate that a root, although itself still a linear reversible transformation, can require asymptotically deeper circuit implementations than the original unitary.

\end{abstract}

\maketitle
\tableofcontents

\section{Introduction}
Galois Field arithmetic is employed by a variety of quantum algorithms, such as Decoded Quantum Interferometry \cite{jordan2024optimization}, proofs of quantumness by cryptographic protocols (such as \cite{kahanamoku2022classically}), and the discrete logarithm problem over the elliptic curve group \cite{proos2003shor}.  As such, it is important to optimize circuit implementations corresponding to the basic field operations. 

Galois Field GF$(p^m)$, where $p$ is prime and $m$ is a positive integer is a finite field whose elements may be conveniently viewed as polynomials $a_0 + a_1x + \cdots + a_{m-1}x^{m-1}$, where $a_0, a_1, \ldots, a_{m-1} \in \mathbb{F}_p$.  GF$(p^m)$ addition is defined as component-wise addition with coefficients reduced over $\mathbb{F}_p$, and multiplication is defined as polynomial multiplication modulo an irreducible polynomial of degree $m$.  The selection of a specific irreducible polynomial affects the representation of the given field, but all fields of a fixed size are isomorphic. 

We focus on arithmetic in the field GF$(2^m)$ (i.e. $p{=}2$), which is a popular choice in applications.  This is because the selection of $p{=}2$ does not affect the difficulty of most complex problems of value (such as discrete logarithm), while the mapping of field elements to bits/qubits is particularly natural, and the cost of arithmetic circuits is low.  Specifically, In-place addition $\ket{a,b}\mapsto \ket{a, a \oplus b}$, where $a, b \in \text{GF}(2^m)$ can be accomplished by a depth-1 layer of $m$ $\cnotgate$ gates via component-wise EXOR, and the out-of-place addition $\ket{a,b,c}\mapsto \ket{a, b, c\oplus a \oplus b}$ can be accomplished by a depth-2 circuit with $2n$ $\cnotgate$ gates.  The other field operations (i.e., multiplication and division) are not as simple. In this paper, we focus on ancilla-free implementation of field multiplication and low ancilla constructions of division and base our work on previous studies such as \cite{van2019space, banegas2020concrete, kim2024toffoli, putranto2023depth}. We consider constructions that use a large number of ancilla, such as \cite{vandaele2025quantumbinaryfieldmultiplication}, to be out of scope due to the additional cost of ancilla. Ancilla-free constructions for multiplication have varying asymptotic $\toffoligate$ cost ranging from $O(m2^{\log_2^*{m}})$ to $O(m^2)$ but share the same $O(m^2)$ $\cnotgate$ cost, though the $\cnotgate$ cost is usually lower in practice. 

In this work we show that ancilla-free Karatsuba-based constructions can be optimized to have $O(m^{\log_2{3}})$ $\cnotgate$ cost, matching its $\toffoligate$ cost by carefully choosing an irreducible polynomial such that the multiplication by the constant $1+x^{\lceil\frac{m}{2}\rceil}$ \cite{van2019space} can be implemented with $O(m)$ $\cnotgate$ gates. We show that the reduction in the count of $\cnotgate$ is not just asymptotic but also practical, as can be seen in \tab{multiplication} and \tab{division}. We also reduce the $\toffoligate$ cost simultaneously by using a more general version of the Karatsuba formula known as the Karatsuba-like formula, first introduced in \cite{Weimerskirch2006GeneralizationsOT}.

For GF$(2^m)$ division, we use addition chains to compute the multiplicative inverse, which require fewer multiplications (and hence fewer ancilla) than the addition chain implicitly defined by the standard Itoh-Tsujii algorithm. This reduction, in addition to the intrinsic reduction we get from using our optimized multiplication circuits, leads to up to $28\%$ decline in both $\toffoligate$ and $\cnotgate$ gate costs for cryptographically relevant field sizes.

Since we study reversible circuits composed of the $\cnotgate$ and $\toffoligate$ gates. We focus, primarily, on minimizing the cost function $10 * \toffoligate + \cnotgate$ using no ancilla for multiplication and without increasing the typical ancilla count for division.  This is due to the combination of the sizes of the field considered and the relative costs of the $\cnotgate$ and $\toffoligate$ gates \cite{huggins2025fluid}.

\section{Literature Review}

GF($2^m$) multiplication circuits can be constructed in several ways with ancilla, such as \cite{roetteler2013discretelogarithms, kepley2015quantum, vandaele2025quantumbinaryfieldmultiplication}, as well as without ancilla. The ancilla-free constructions fall broadly into three categories: long multiplication-based constructions such as \cite{mastrovito1988vlsi, maslov2009m2}, Karatsuba-based constructions such as \cite{van2019space, putranto2023depth}, and CRT-based (Chinese Remainder Theorem) constructions such as \cite{kim2024toffoli}. 

We refer the reader to \cite{JangSrivastavaBaksi} for the previous best ancilla-free Karatusba-based constructions; here we provide a short summary of the state of the art.

The Karatsuba-based algorithm \cite{van2019space} has a $\toffoligate$ count that scales as $O(m^{\log_2{3}})$ and $O(m^2)$ $\cnotgate$ count.  Here, $O(m^2)$ comes from costing modular multiplication (and division) by the constant $1+x^{\lceil\frac{m}{2}\rceil}$, whereas all other components of the construction have $O(m^{\log_2{3}})$ $\cnotgate$ count or less. Putranto et al. \cite{putranto2023depth} attempted a constant optimization by removing the division by $1+x^{\lceil\frac{m}{2}\rceil}$ step and reordering the other operations.  This, however, only works if the target register is $\ket{0}$; otherwise, their construction does the mapping $\ket{a,b,c} \mapsto \ket{a,b,c(1+x^{\lceil\frac{m}{2}\rceil})+ab}$. Accepting this construction as an out-of-place multiplier, we were still unable to reproduce the $\cnotgate$ numbers in their tables; our numbers, however, are consistent with the numbers that \cite{JangSrivastavaBaksi} report for this construction and thus we will refer to their estimates instead. As we show in this work, multiplication/division by the constant $1 + x^{\lceil\frac{m}{2}\rceil}$ can be cheap, so we can make this construction in-place.


The CRT-based construction from \cite{kim2024toffoli}  has the best scaling of $\toffoligate$ gates at $O(m\cdot 2^{\log_2^*{m}})$ but it comes with a large $\cnotgate$ count, which, while asymptotically same as the Karatsuba-based construction at $O(m^2)$, is several multiples of the Karatsuba-based methods count.  As a result, the $\cnotgate$ gate count dominates the cost of the implementation, and we consider this scenario impractical.  This method has four parts, two of which require $\cnotgate$ circuits with a low $\cnotgate$ count at $O(m\log(m))$, but the other two require the $\cnotgate$ gate count of $O(m^2)$. 

Another approach to gate count optimization was proposed in an earlier work \cite{kepley2015quantum}, which gives the desired $O(m^{\log_2{3}})$ gate complexity only when the irreducible polynomial is a trinomial.  This, however, covers only some but not all cases \cite{seroussi1998table}.  In contrast, our constructions work for all $m$ and cover both field multiplication and division operations, which is an important consideration for applications, as they tend to rely on both these operations simultaneously.

In this work, we focus on the Karatsuba-based methods, although many of the techniques we introduce can optimize the CRT-based circuits as well. We further highlight that Karatsuba multiplication over GF$(2^m)$ results in asymptotically smaller and more practical circuits compared to a standard Mastrovito multiplier \cite{mastrovito1988vlsi}, which can be thought of as an analogue of long multiplication for the Galois Field.  This holds even for small Galois Field sizes, but, interestingly, does not appear to apply to integers, despite a best approach to optimize the gate counts and ancillae \cite{gidney2019asymptotically}.



Division $a/b$ is done by first computing the multiplicative inverse $b^{-1}$, obtaining the result as the multiplication $a \cdot b^{-1}$, and finally uncomputing $b^{-1}$. Two approaches to computing the multiplicative inverse were proposed and are analogous to classical algorithms for computing the multiplicative inverse; the first relies on Fermat's Little Theorem (FLT) \cite{itoh1988fast} and the second is the Extended Euclidean Algorithm \cite{BY_GCD}. Jang et al. \cite{banegas2020concrete} studied both approaches, and Table 2 therein compares the cost of each approach for different field sizes.  We chose to focus on the FLT-based approach based on the study \cite[Table 2]{banegas2020concrete}, although it is worth noting that research continues to improve and optimize both approaches (see \cite{JangSrivastavaBaksi}).

\section{GF($2^m$) Multiplication And Division}

\subsection{Multiplication}
For GF$(2^m)$ multiplication, we start with the construction of \cite{van2019space}.  It utilizes Karatsuba algorithm \cite{karatsuba1962multiplication} to multiply unreduced polynomials over the binary field $\mathbb{F}_2$ with up to $m/2$ digits, and then, rather than introducing $m$ ancilla qubits, employs a smart trick \cite[Table 2]{van2019space} to multiply and reduce polynomials in-place during one last iteration of Karatsuba's algorithm.  Not only does it allow for obtaining a no-ancilla implementation of GF multiplication, but it also reduces the leading constant in front of the complexity figure by a factor of $2$ against a straightforward approach utilizing ancillae.  As a result,  the $\toffoligate$ count for $m\,{=}\,2^k$ is precisely $3^k$.  

Note that \cite{van2019space} performs the transformation $\ket{a,b,0} \mapsto \ket{a,b,ab}$, which we stay consistent with, and note that turning it into a true out-of-place mapping $\ket{a,b,c} \mapsto \ket{a,b,c\oplus ab}$ can be done by applying a $\cnotgate$ circuit to the input side that applies the transformation $\ket{c} \mapsto \ket{c x^{-n}} $ which takes $O(mw)$ $\cnotgate$ gates where $w$ is the weight of the irreducible polynomial.  In our case, $w$ remains small and does not affect asymptotics or even the $\cnotgate$ gate count in a meaningful way.


The bottleneck of this construction is two $\cnotgate$ gate circuits, the first multiplies by $1+x^{\lceil{m/2}\rceil}$ and the second divides by it (and it is thus the inverse of the first), which \cite{van2019space} synthesize using LUP, requiring $O(m^2)$ gates in the worst case.  We show how to reduce the cost of the multiplication by the constant polynomial $1+x^{\lceil{m/2}\rceil}$ from $O(m^2)$ down to either $O(m)$ or $O(m \log(m))$, depending on the scenario of interest, by first showing that irreducible polynomials exist for which this can be done and then offering constructive algorithms that guarantee these costs, thus making this bottleneck cost become irrelevant to the overall complexity since the other operations scale as $O\left(m^{\log_2{3}}\right)$.  The $O(m)$ implementation is to be used when only field multiplication operation is required, while the $O(m\log(m))$ implementation is to be used when both multiplication and division are needed, since the irreducible polynomials that enable this construction also enable squaring with $O(m\log(m))$ gates (see \sec{div}), required for efficient implementation of the field division operation.


Moving from ``pure'' Karatsuba, we examined two generalizations of Karatsuba, namely the Toom-Cook strategy (see Appendix \ref{appendix:ancilla_free_toom_cook}) and Karatsuba-like formulas.  To our knowledge, we offer a first circuit implementation of a quantum Toom-Cook algorithm for $GF(2^m)$ multiplication that relies on a known classical trick due to Zimmermann and Quercia \cite{zimmermann_faster_multiplication}. Unfortunately, this construction did not improve the cost over Karatsuba-like formulas also considered and reported in our work. We examined Karatsuba-like formulas that were first introduced classically in \cite{Weimerskirch2006GeneralizationsOT} and were recently used by Kim et al. \cite{kim2024toffoli} as a base case ($m \leq 8$). In contrast, we use it for the full construction, which led to improved $\toffoligate$ and $\cnotgate$ counts.  We further improved our $\toffoligate$ cost using a global optimization that we introduce in \sec{pctof} and our $\cnotgate$ cost using both local and global optimizations.

\subsection{Division}
A promising approach for implementing GF division \cite{itoh1988fast} utilizes Fermat's Little Theorem.  To compute $a/b$, where $a,b \in \text{GF}(2^m)$ we look into inverting $b$ and multiplying it with $a$.  To find $b^{-1}$, rely on the Fermat's Little Theorem, $b^{2^m-1}= 1$, and rewrite this equation as $b^{-1}=b^{2^m-2}$.  Thus, to find the desired $b^{-1}$ it suffices to square $b^{2^{m-1}-1}$.  To find the power $2^{m-1}{-}1$ of an element $b$ one can utilize a square and multiply approach, where, due to the binary expansion of $2^{m-1}{-}1$ consisting of just $1$'s, only $O(\log(m))$ multiplications suffice, together with $O(m)$ squaring operations in a set of $O(\log(m))$ batches.  The squaring operation in $\text{GF}(2^m)$ is a linear reversible transformation.  The cost of implementing linear reversible operations tends to be low.  Thus, it is reasonable to expect that the cost of the field division operation when implemented by the Itoh-Tsujii technique will be dominated by the cost of multiplications rather than squaring operations.  The work \cite[Table 2]{banegas2020concrete} investigates this and shows that relying on the Itoh-Tsujii algorithm results in better circuits than the alternatives.  However, due to the use of quadratic complexity implementation of linear reversible squaring circuits and their powers, the $\cnotgate$ gate overhead is significant.  Indeed, its complexity scales as $O(m^2 \log(m))$, overwhelming the $\toffoligate$ count complexity of $O(m^{\log_2{3}}\log(m))$, and numeric data shows that the $\cnotgate$ count significantly exceeds the $\toffoligate$ count.  Here, our goal is to reduce this overhead. 

GF$(2^m)$ has a useful feature that it allows to reversibly square an element $a$ in-place, $C: \ket{a} \mapsto \ket{a^2}$, and the circuit $C$ accomplishing this is linear reversible.  This is because $a^2=a_0+a_1x^2+\cdots+a_{m-1}x^{2(m-1)}$, and the reduction of monomials modulo the irreducible polynomial is a linear Boolean function. To construct the linear reversible operation performing the transformation $\ket{a} \,{\mapsto}\, \ket{a^2}$, it suffices to build the linear $m\,{\times}\, m$ map where columns show the mapping of monomials under the GF squaring operation.  This implies that for any $a \in \text{GF}(2^m)$ and positive integer $k$, $a^{2^k}$ is a linear reversible circuit, and it can be considered simple to compute \cite{patel2008optimal}.  Note, however, that \cite{patel2008optimal} offers an asymptotically optimal construction with complexity $O(m^2/\log(m))$ that, due to the high leading constant, can be impractical. 

Efficient squaring enables the construction of an efficient circuit for the last missing piece, the division/inversion (subtraction and addition in GF$(2^m)$ are identical).  The construction due to Itoh and Tsujii \cite{itoh1989structure, takagi2001fast}, also discussed in \cite{amento2012efficient, banegas2020concrete} in the context of reversible/quantum computations, allows to implement field division at the cost of $O(\log(m))$ multiplications and $O(\log(m))$ batches of repeated squaring operations.  Our construction of efficient multiplier and squaring circuits reduces the cost of the field division operation.  The details, including numeric results, can be found in \sec{div}.

Järvinen and Dimitrov \cite{dimitrov2013another} optimize the Itoh-Tsujii algorithm by employing the addition chains.  This results in the reliance on fewer GF multiplications necessary to implement the division. The addition chains idea has been studied extensively in classical computing \cite{Bernstein2025}.  Quantumly, however, it remains relatively underexplored, although Taguchi and Takayasu \cite{TaguchiTakayasu} used it to reduce the space required for inversion at the cost of slightly increasing the $\toffoligate$ gate count and in a later paper \cite{Taguchi2024} they replaced the karatsuba multiplier with the CRT multiplier \cite{kim2024toffoli} that allowed them to reduce the $\toffoligate$ gate count.  In our work, we use both the standard Itoh-Tsujii algorithm and its modification by the addition chains \cite{dimitrov2013another} to minimize the $\toffoligate$ and $\cnotgate$ gate counts, and sometimes ancillae.

\section{Multiplication By The constant $1 + x^{\lceil\frac{m}{2}\rceil}$}
We first describe the structure of the linear invertible matrix computing GF$(2^m)$ multiplication by a constant polynomial $1+x^{\lceil{m/2}\rceil}$. A GF$(2^m)$ element is represented by a column vector of length $m$ with the top-to-bottom entries corresponding to the monomial powers of $0,1,\cdots,m{-}1$.  The matrix multiplication applies to the left of the column vector. 

The first column of the matrix $M$ encoding the multiplication by the constant polynomial $1+x^{\lceil{m/2}\rceil}$ is vector $(1,0,0,\cdots,0,1,0,0,\cdots,0)^T$ with ones in positions $0$ and $\lceil{m/2}\rceil$. It encodes the result of the multiplication of $1+x^{\lceil{m/2}\rceil}$ by $1$. The next column is $(0,1,0,0,\cdots,0,1,0,0,\cdots,0)^T$, i.e., it is a cyclic shift of the first column by one position down.  The first $m{-}\lceil{m/2}\rceil{-}1$ columns are all shifts by one of the previous column.  When $c=m-\lceil{m/2}\rceil$, the $c^\text{th}$ column must represent $x^c+x^m$, but $x^m$ is reduced modulo the irreducible polynomial $x^m+x^{l_1}+x^{l_2}+\cdots+x^{l_k}+1$, where $m>l_1>l_2>\cdots>l_k>0$, and is thus replaced with $x^{l_1}+x^{l_2}+\cdots+x^{l_k}+1$. Here, we consider irreducible polynomials where $l_1 < m-\lceil{m/2}\rceil$. This means that the new column, $(1,0,0,\cdots,0,1,0,0,\cdots,0,\cdots1,0,0,\cdots,0)^T$ is being introduced with $1$'s at positions $0, l_1, l_2,\cdots,l_k,$ and $c$. This column is then shifted by $1$ for each next position to complete the matrix $M$. 

Recall that a $\cnotgate$ circuit can implement a matrix $M$ by performing column and row addition operations. These matrix operations correspond to assigning the $\cnotgate$ gates to the beginning or the end of the circuit.

\begin{lemma}
Subject to a straightforward assumption, in GF$(2^m)$ multiplication by a constant polynomial $1+x^{\lceil{m/2}\rceil}$ can be implemented by a $\cnotgate$ circuit with $O(m \log(m))$ gates.
\end{lemma}
\begin{proof}
The assumption we rely on is that there always exist an irreducible polynomial $x^m+x^{l_1}+x^{l_2}+\cdots+x^{l_k}+1$, where $n>l_1>l_2>\cdots>l_k>0$, such that either $l_1 \in O(\log(m))$ or $l_1{-}l_k \in O(1)$.  Indeed, it is known \cite{chebolu2011counting} that the probability of drawing an irreducible polynomial from the uniform distribution of all polynomials of degree $m$ is $1/m$.  Thus, there are linearly many irreducible polynomial candidates for either $l_1 \sim \log(m)$ or $l_1{-}l_k \sim Const$ case, offering a constant probability to find an irreducible polynomial.  This probability can be boosted to almost 1 by increasing the parameters only slightly and without affecting the asymptotics.

\subsection{Trinomials with Even $m$}
We first consider a simple case of trinomial irreducible polynomials $x^m{+}x^l{+}1$ and even $m\,{=}\,2n$. In all such cases, the cost of the multiplication by the constant polynomial $1+x^{\lceil{m/2}\rceil}$ is $O(m)$, and no special irreducible polynomial needs to be constructed. \alg{1} constructs the $\cnotgate$ circuit with $6n{-}l=
3m{-}l = O(m)$ $\cnotgate$ gates.  The number of $\cnotgate$ gates is further reducible to $1.5m{-}l$ if it suffices to output the bits in any order, which can be a useful trick for an all-to-all connected machine.

{\centering
\begin{algorithm}[H]
\caption{$\cnotgate$ circuit construction for GF$(2^m)$ with even $m=2n$ and trinomial $x^m+x^l+1$, where $n>l$}\label{alg:1}
\begin{algorithmic}[1]	
    \State{Input: matrix $M=\{m_{i,j}\}_{i,j=1..2n}$ of size $m{\times}m$}
    \State{Output: circuit reducing matrix $M$ to the identity using $O(m)$ $\cnotgate$ gates}
    \For{$i=0$ to $n{-}1$} \Comment{$n$ $\cnotgate$ gates}
        \State{Add row $i$ to row $n+i$}
    \EndFor
    \For{$i=0$ to $n{-}1$} \Comment{$n$ $\cnotgate$ gates}
        \State{Add column $i$ to column $n+i$}
    \EndFor
    \For{$i=l$ to $n{-}1$} \Comment{$n{-}l$ $\cnotgate$ gates}
        \State{Add row $n+i$ to row $i$}
    \EndFor 
    \State{Cyclically shift bottom $n$ rows $l$ positions up} \Comment{$n{-}1$ $\swapgate$ gates}
\end{algorithmic}
\end{algorithm}
}

The workings of \alg{1} are best illustrated with an example. Next, we consider multiplication by $1{+}x^5$ over GF$(2^{10})$ with the irreducible polynomial $x^{10}{+}x^3{+}1$, where bold font is used to highlight the bit values that changed:

$
M =
\left[
\begin{array}{ccccc|ccccc}
1 & 0 & 0 & 0 & 0 & 1 & 0 & 0 & 0 & 0 \\
0 & 1 & 0 & 0 & 0 & 0 & 1 & 0 & 0 & 0 \\
0 & 0 & 1 & 0 & 0 & 0 & 0 & 1 & 0 & 0 \\
0 & 0 & 0 & 1 & 0 & 1 & 0 & 0 & 1 & 0 \\
0 & 0 & 0 & 0 & 1 & 0 & 1 & 0 & 0 & 1 \\ \hline
1 & 0 & 0 & 0 & 0 & 1 & 0 & 1 & 0 & 0 \\
0 & 1 & 0 & 0 & 0 & 0 & 1 & 0 & 1 & 0 \\
0 & 0 & 1 & 0 & 0 & 0 & 0 & 1 & 0 & 1 \\
0 & 0 & 0 & 1 & 0 & 0 & 0 & 0 & 1 & 0 \\
0 & 0 & 0 & 0 & 1 & 0 & 0 & 0 & 0 & 1 \\
\end{array}
\right]
\;{\stackrel{{\mbox{3-5:}}}{\mapsto}}\;
\left[
\begin{array}{ccccc|ccccc}
1 & 0 & 0 & 0 & 0 & 1 & 0 & 0 & 0 & 0 \\
0 & 1 & 0 & 0 & 0 & 0 & 1 & 0 & 0 & 0 \\
0 & 0 & 1 & 0 & 0 & 0 & 0 & 1 & 0 & 0 \\
0 & 0 & 0 & 1 & 0 & 1 & 0 & 0 & 1 & 0 \\
0 & 0 & 0 & 0 & 1 & 0 & 1 & 0 & 0 & 1 \\ \hline
{\bf 0} & 0 & 0 & 0 & 0 & {\bf 0} & 0 & 1 & 0 & 0 \\
0 & {\bf 0} & 0 & 0 & 0 & 0 & {\bf 0} & 0 & 1 & 0 \\
0 & 0 & {\bf 0} & 0 & 0 & 0 & 0 & {\bf 0} & 0 & 1 \\
0 & 0 & 0 & {\bf 0} & 0 & {\bf 1} & 0 & 0 & {\bf 0} & 0 \\
0 & 0 & 0 & 0 & {\bf 0} & 0 & {\bf 1} & 0 & 0 & {\bf 0} \\
\end{array}
\right]
\;{\stackrel{{\mbox{6-8:}}}{\mapsto}}\;
\left[
\begin{array}{ccccc|ccccc}
1 & 0 & 0 & 0 & 0 & {\bf 0} & 0 & 0 & 0 & 0 \\
0 & 1 & 0 & 0 & 0 & 0 & {\bf 0} & 0 & 0 & 0 \\
0 & 0 & 1 & 0 & 0 & 0 & 0 & {\bf 0} & 0 & 0 \\
0 & 0 & 0 & 1 & 0 & 1 & 0 & 0 & {\bf 0} & 0 \\
0 & 0 & 0 & 0 & 1 & 0 & 1 & 0 & 0 & {\bf 0} \\ \hline
0 & 0 & 0 & 0 & 0 & 0 & 0 & 1 & 0 & 0 \\ 
0 & 0 & 0 & 0 & 0 & 0 & 0 & 0 & 1 & 0 \\
0 & 0 & 0 & 0 & 0 & 0 & 0 & 0 & 0 & 1 \\
0 & 0 & 0 & 0 & 0 & 1 & 0 & 0 & 0 & 0 \\
0 & 0 & 0 & 0 & 0 & 0 & 1 & 0 & 0 & 0 \\
\end{array}
\right] \linebreak
\;{\stackrel{{\mbox{9-11:}}}{\mapsto}}\;
\left[
\begin{array}{ccccc|ccccc}
1 & 0 & 0 & 0 & 0 & 0 & 0 & 0 & 0 & 0 \\
0 & 1 & 0 & 0 & 0 & 0 & 0 & 0 & 0 & 0 \\
0 & 0 & 1 & 0 & 0 & 0 & 0 & 0 & 0 & 0 \\
0 & 0 & 0 & 1 & 0 & {\bf 0} & 0 & 0 & 0 & 0 \\
0 & 0 & 0 & 0 & 1 & 0 & {\bf 0} & 0 & 0 & 0 \\ \hline
0 & 0 & 0 & 0 & 0 & 0 & 0 & 1 & 0 & 0 \\ 
0 & 0 & 0 & 0 & 0 & 0 & 0 & 0 & 1 & 0 \\
0 & 0 & 0 & 0 & 0 & 0 & 0 & 0 & 0 & 1 \\
0 & 0 & 0 & 0 & 0 & 1 & 0 & 0 & 0 & 0 \\
0 & 0 & 0 & 0 & 0 & 0 & 1 & 0 & 0 & 0 \\
\end{array}
\right]
\;{\stackrel{{\mbox{12:}}}{\mapsto}} \; Id.
$
\end{proof}
Unfortunately, this simple algorithm does not work for other cases because the matrix can no longer be divided into four quadrants of equal size that can be quickly cleaned (i.e., for odd $m$) or the irreducible polynomial has more than three terms and there are more off-diagonal $1$s to clean.  We will first focus on developing an algorithm for even $m$ and any number of polynomial terms, then we turn to the case of odd $m$, in which case the algorithm follows the same path except that the $n{\times}n$ quadrants are contained to corners, cyclic shift of the circulant matrix (explained below) is not needed, and the dimensions of the transformations are slightly different.

\subsection{General Irreducible Polynomial with Even $m$}

\alg{2} gives a construction for multiplication by the constant $1{+}x^{\lceil m/2 \rceil}$, Notice that \alg{2} contains every step in \alg{1} with a few extra steps.

{\centering
\begin{algorithm}[H]
\caption{$\cnotgate$ circuit construction for GF$(2^m)$ with even $m=2n$ and irreducible polynomial $x^m+x^{l_1}+x^{l_2}+\cdots+x^{l_k}+1$, where $n>l_1>l_2>\cdots>l_k$} \label{alg:2}
\begin{algorithmic}[1]	
    \State{Input: matrix $M=\{m_{i,j}\}_{i,j=0..2n-1}$ of size $m{\times}m$, $x^m+x^{l_1}+x^{l_2}+\cdots+x^{l_k}+1$ is irreducible polynomial}
    \State{Output: circuit reducing matrix $M$ to the identity using $O(m)$ $\cnotgate$ gates}
    \For{$i{=}0$ to $n{-}1$}  \Comment{$n$ $\cnotgate$ gates}
        \State{Add row $i$ to row $n{+}i$}
    \EndFor
    \For{$i{=}0$ to $n{-}1$}  \Comment{$n$ $\cnotgate$ gates}
        \State{Add row $n{+}i$ to row $i$}
    \EndFor
    \For{each non-zero element $M_{i,j}$ where $i=0..n{-}1, j=n..2n{-1}$}  \Comment{$n + \sum l_i$ $\cnotgate$ gates}
        \State{Add column $i$ to column $j$}
    \EndFor
    \State{Cyclically shift bottom $n$ rows $l_k$ positions up} \Comment{$n{-}1$ $\swapgate$ gates}
    \For{$j=n$ to $m-1$}  \Comment{$n(k{-}1)$ $\cnotgate$ gates}
        \For{$i{=}1$ to $k{-}1$}  
            \State{Add row $j$ to row $j+l_i{-}l_k$ when $j+l_i{-}l_k < m$}
        \EndFor
    \EndFor
    \State{Diagonalize rows $n..2n-l_1+l_k$ by adding columns $n..2n-l_1+l_k$ containing standard basis vectors as necessary, diagonalize bottom right $l_1{-}l_k {\times} l_1{-}l_k$ corner by Gaussian elimination over the last $l_1{-}l_k$ rows/columns. \Comment{$n(l_1-l_k)$ $\cnotgate$ gates}}
\end{algorithmic}
\end{algorithm}
}
The original matrix $M$ has the form 
$M = \begin{bmatrix}
    I_n & A \\
    I_n & B
\end{bmatrix}$, where $I_n$ is the $n {\times} n$ identity matrix.
\alg{2} performs the following transformations:
\begin{itemize}
    \item[3-5:] Add rows of top $I_n$ to the bottom $I_n$ to reduce the bottom $I_n$ matrix to all-zero matrix $0_n$. This costs $n$ $\cnotgate$ gates.  The leftover matrix is $\begin{bmatrix}
    I_n & A \\
    0_n & A{+}B
    \end{bmatrix}$, where $A{+}B$ is the element-wise sum of $A$ and $B$.
    \item[6-8:] Add bottom $n$ rows to the top $n$ rows.  This costs $n$ $\cnotgate$ gates. The leftover matrix is $\begin{bmatrix}
    I_n & B \\
    0_n & A{+}B
    \end{bmatrix}$.
    \item[9-11:] Use columns of $I_n$ to reduce $k{+}1$ non-zero diagonal and superdiagonals in the matrix $B$ to zero. This uses $n$ $\cnotgate$ gates for the main diagonal and at most $\sum_{i=1}^k l_i \leq l_1k$ $\cnotgate$ gates to reduce the top left upper triangular corner with the sides of size $l_1$. The leftover matrix is $\begin{bmatrix}
    I_n & 0_n \\
    0_n & A{+}B
    \end{bmatrix}$.
    \item[11:] We note that the matrix $A{+}B$ is circulant.  This step is only necessary for even $m$ and can be ignored for odd $m$. The goal is to put 1's into the main diagonal of the matrix $A+B$ by cyclically shifting its columns. The cost is $3n{-}3$ $\cnotgate$ gates ($n{-}1$ $\swapgate$s). The leftover matrix has the form $\begin{bmatrix}
    I_n & 0_n \\
    0_n & (A{+}B)^\uparrow
    \end{bmatrix}$.
    \item[13-17:] This step reduces $(A{+}B)^\uparrow$ to the following form, 
    $\begin{bmatrix}
    I_{n-l_1+l_k} & C_{n-l_1+l_k \times l_1-l_k} \\
    0_{l_1-l_k \times n-l_1+l_k} & D_{l_1-l_k}
    \end{bmatrix}$, where $C_{n-l_1+l_k \times l_1-l_k}$ and $D_{l_1-l_k}$ are some $n-l_1+l_k \times l_1-l_k$ and $l_1-l_k \times l_1-l_k$ matrices, correspondingly, and $D_{l_1-l_k}$ is furthermore invertible, and $0_{l_1-l_k \times n-l_1+l_k}$ is a proper size all-zero matrix.  All other parts of the leftover matrix ($0_n$ and $I_n$) are kept unchanged.  This step uses no more than $n(k{-}1)$ $\cnotgate$ gates.
    \item[18:] $C_{n-l_1+l_k \times l_1-l_k}$ is reduced to zero using column additions from matrix $I_{n-l_1+l_k}$.  And then the invertible square matrix $D_{l_1-l_k}$ is diagonalized by standard Gaussian elimination.  The $\cnotgate$ cost is at most $(l_1-l_k)(n-l_1+l_k)+(l_1-l_k)^2 = n(l_1-l_k)$.
\end{itemize}
The full implementation thus contains no more than $n(k{+}l_1{-}l_k{+}5) + l_1k -3$ $\cnotgate$ gates, which is in $O(m)$ so long as $l_1{-}l_k \in O(1)$. Recall that $k{+}2$ is the weight of the irreducible polynomial, $l_1$ is the second largest power of the monomial in it (largest power is $m$), $l_k$ is the second smallest power of the monomial in it (smallest power is $0$, corresponding to the constant term), and $n=m/2$.  This upper bound requires that the irreducible polynomial be found among those with a low $l_1{-}l_k$.  There are ${\Theta}(n2^{l_1{-}l_k})$ different irreducible polynomial candidates for each value $l_1{-}l_k$.

\subsection{Trinomials and General Irreducible Polynomials with Odd $m$}

Similar constructions can be used for odd $m \,{=}\, 2n {+} 1$, as shown in \alg{3} and \alg{4}. \alg{3} is essentially \alg{1} with a few extra steps leading to a total $\cnotgate$ count of $n(l{+}6) - 3l = O(m \log(m))$ for trinomials, while \alg{4}, which is a modified version of \alg{2}, has a count of $\cnotgate$ of $n(2l_1 {+}5) - l_1^2 - l_1 + \sum l_i = O(m \log(m))$. Both constructions differ from their counterparts for even $m$ in the same way, namely, the replacement of the cyclic shift step with a few row or column operations that depend on $l_i$.

{\centering
\begin{algorithm}[H]
\caption{$\cnotgate$ circuit construction for GF$(2^m)$ with odd $m=2n+1$ and trinomial $x^m+x^l+1$, where $n>l$}\label{alg:3}
\begin{algorithmic}[1]	
    \State{Input: matrix $M=\{m_{i,j}\}_{i,j=1..2n+1}$ of size $m{\times}m$}
    \State{Output: circuit reducing matrix $M$ to the identity using $O(m \log(m))$ $\cnotgate$ gates}
    \For{$i=0$ to $n{-}1$} \Comment{$n$ $\cnotgate$ gates}
        \State{Add row ${i}$ to row $n{+}i{+}1$}
    \EndFor
    \For{$i=0$ to $n{-}1$} \Comment{$n$ $\cnotgate$ gates}
        \State{Add column $i$ to column $n{+}i$}
    \EndFor
    \For{$i=l$ to $n{-}1$} \Comment{$n{-}l$ $\cnotgate$ gates}
        \State{Add column $i$ to column $n{-}l{+}i$}
    \EndFor 
    \For{$i=0$ to $n{-}l{-}1$} \Comment{$2(n{-}l)$ $\cnotgate$ gates}
        \State{Add row $n+i$ to row $n+i+1$}
        \State{Add row $n+i$ to row $n+i+l+1$}
    \EndFor 
    \State{Diagonalize rows $n..2n-l-1$ by adding columns $n..2n-l-1$ containing standard basis vectors as necessary, diagonalize bottom right $(l{+}1) {\times} (l{+}1)$ corner by Gaussian elimination over the last $l{+}1$ rows/columns. \Comment{$n(l{+}1)$ $\cnotgate$ gates}}
\end{algorithmic}
\end{algorithm}
}

{\centering
\begin{algorithm}[H]
\caption{$\cnotgate$ circuit construction for GF$(2^m)$ with odd $m=2n+1$ and irreducible polynomial $x^m+x^{l_1}+x^{l_2}+\cdots+x^{l_k}+1$, where $n>l_1>l_2>\cdots>l_k$} \label{alg:4}
\begin{algorithmic}[1]	
    \State{Input: matrix $M=\{m_{i,j}\}_{i,j=0..2n-1}$ of size $m{\times}m$, $x^m+x^{l_1}+x^{l_2}+\cdots+x^{l_k}+1$ is irreducible polynomial}
    \State{Output: circuit reducing matrix $M$ to the identity using $O(m \log(m))$ $\cnotgate$ gates}
    \For{$i{=}0$ to $n{-}1$}  \Comment{$n$ $\cnotgate$ gates}
        \State{Add row $i$ to row $n{+}i{+}1$}
    \EndFor
    \For{$i{=}0$ to $n{-}1$}  \Comment{$n$ $\cnotgate$ gates}
        \State{Add row $n{+}i{+}1$ to row $i$}
    \EndFor
    \For{each non-zero element $M_{i,j}$ where $i=0..n{-}1, j=n..2n$}  \Comment{$n + \sum l_i$ $\cnotgate$ gates}
        \State{Add column $i$ to column $j$}
    \EndFor
    \For{$j{=}n$ to $m{-}l_1{-}2$}  \Comment{$(n{-}l_1)(k {+} 1)$ $\cnotgate$ gates}
        \For{$i{=}1$ to $k$} 
            \State{Add row $j$ to row $j{+}l_i{+}1$}
        \EndFor
        \State{Add row $j$ to row $j{+}1$}
    \EndFor
    \State{Diagonalize rows $n..2n{-}l_1{-}1$ by adding columns $n..2n{-}l_1{-}1$ containing standard basis vectors as necessary, diagonalize bottom right $(l_1{+}1) {\times} (l_1{+}1)$ corner by Gaussian elimination over the last $l_1{+}1$ rows/columns. \Comment{$n(l_1{+}1)$ $\cnotgate$ gates}}
\end{algorithmic}
\end{algorithm}
}

\subsection{Irreducible Polynomials with Odd $m$, Circuit Cost $O(m)$.}

Here, we select to work with the irreducible polynomial $x^m + (x^{n-1} + x^{n-2} + ... + x^{n-l_1}) + (x^{l_2} + x^{l_2-1} + ... + 1)$, where $n{-}l_1 \,{>}\, l_2$ and $l_1,l_2 \geq 0$. We first develop a circuit with $O(m)$ $\cnotgate$ gates to reduce the linear reversible matrix corresponding to the field multiplication by $x^{n+1}{+}1$ to the problem of the synthesis of a circulant $(n{+}1)\,{\times}\,(n{+}1)$ matrix such that $\text{Dist}(row_i,row_{i+1})\,{=}\,2$, and then implement the circulant matrix using $O(m)$ gates.

\alg{5} reduces the $m\,{\times}\,m$ matrix $M$ into the form $\begin{bmatrix}
    I_{n} & 0_{n+1 \times n} \\
    0_{n \times n+1} & C_{n+1 \times n+1}
    \end{bmatrix}$,
where $C_{n+1 \times n+1}$ is circulant with ones in the first column at positions $0, 1, ..., l_2{+}1, n{-}l_1{+}1, n{-}l_1{+}2, ..., n$ using $O(m)$ $\cnotgate$ gates.

{\centering
\begin{algorithm}[H]
\caption{The reduction of the multiplication by $x^{n+1}+1$ for odd $m=2n+1$ and irreducible polynomial $x^m + (x^{n-1} + x^{n-2} + ... + x^{n-l_1}) + (x^{l_2} + x^{l_2-1} + ... + 1)$, where $n-l_1 > l_2$} \label{alg:5}
\begin{algorithmic}[1]	
    \State{Input: matrix $M=\{m_{i,j}\}_{i,j=0..2n-1}$ of size $m{\times}m$}
    \State{Output: matrix $M'=\begin{bmatrix}
    I_{n} & 0_{n+1 \times n} \\
    0_{n \times n+1} & C_{n+1 \times n+1}
    \end{bmatrix}$, where $C_{n+1 \times n+1}$ is circulant with $\text{Dist}(row_i,row_{i+1})=2$. The reduction of $M$ to $M'$ is obtained using $O(m)$ $\cnotgate$ gates}
    \For{$i{=}0$ to $n{-}1$}  \Comment{$n$ $\cnotgate$ gates}
        \State{Add row $i$ to row $n{+}i{+}1$}
    \EndFor
    \For{$i{=}0$ to $n{-}1$}  \Comment{$n$ $\cnotgate$ gates}
        \State{Add row $n{+}i{+}1$ to row $i$}
    \EndFor
    \For{$i{=}n{-}1$ downto $0$}  \Comment{$\leq n$ $\cnotgate$ gates}
        \State{{\bf if}(column $n{+}i$ is non-zero) {Add column $n{+}i$ to column $n{+}i{+}1$}}
    \EndFor
    \For{each non-zero element $M_{i,j}$ where $i=0..n{-}1, j=n..m$}  \Comment{$\leq 3n$ $\cnotgate$ gates}
       \State{Add column $i$ to column $j$}
   \EndFor
   \State{Uncompute steps 9-11}  \Comment{$\leq n$ $\cnotgate$ gates}

\end{algorithmic}
\end{algorithm}
}

We next explain the steps that \alg{5} takes:
\begin{itemize}
    \item[3-5:] Reduce $M=\begin{bmatrix}
    I_{n} & A_{n+1 \times n} \\
    B_{n \times n+1} & D_{n+1 \times n+1}
    \end{bmatrix}$, where $B_{n \times n+1}$ is a row of zeroes followed by the $I_{n}$ matrix on the bottom to $M=\begin{bmatrix}
    I_{n} & A_{n+1 \times n} \\
    0_{n \times n+1} & C_{n+1 \times n+1}     
    \end{bmatrix}$. Observe that $C_{n+1 \times n+1}$ is now circulant with $\text{Dist}(row_i,row_{i+1})=2$.
    \item[6-8:] Copy bottom $n$ rows into the top $n$ rows, which has the effect of turning $A_{n+1 \times n}$ into an upper triangular matrix.
    \item[9-11:] Add neighboring columns such that the number of ones in each reduces to no more than $3$---$2$ coming from the consideration that $\text{Dist}(CS^{i}(0^k1^{n-k+1}),CS^{i+1}(0^k1^{n-k+1}))$, where $k,i$ are integers and $CS$ is cyclic shift of the respective Boolean pattern, and $1$ coming from column $i$ having one less $1$ than column $i{+}1$.
     \item[12-14:] Each of at most $3$ ones in the top right matrix is reduced to zero by a column addition from the left top diagonal matrix.
     \item[15:] Return the bottom right matrix to the circulant form. 
\end{itemize}

\alg{6} implements the matrix $C$ using $O(n)$ $\cnotgate$ gates. The reason this algorithm works can be explained as follows:
\begin{itemize}
    \item[3-5:] These transformations reduce the weights of all columns except column $n$ to $2$.
    \item[6:] Since our matrix is full-rank, we can add at most $n$ columns to the rightmost column to force the number of $1$'s in it to be equal to one.
    \item[7-9:] For each column of weight $1$, find the weight-$2$ column it overlaps with and thus use a single column addition to correct the weight of this column to $1$.  The overlap always exists, since if it does not, it implies there is a linear transformation over a subset of qubits described by a matrix whose column weights are all equal to $2$; however, no such matrix can be invertible, leading to a contradiction.
    \item[10:] Fixes the leftover SWAPping transformation. 
\end{itemize}

{\centering
\begin{algorithm}[H]
\caption{Implementation of the circulant matrix $C_{n+1 \times n+1}$ arising in \alg{5}} \label{alg:6}
\begin{algorithmic}[1]	
    \State{Input: matrix $C=\{c_{i,j}\}_{i,j=0..n}$ of size $(n+1){\times}(n+1)$}
    \State{Output: circuit reducing matrix $C$ to the identity using $O(n)$ $\cnotgate$ gates}
    \For{$i{=}1$ to $n$}  \Comment{$n$ $\cnotgate$ gates}
        \State{Add column $i$ to column $i{-}1$}
    \EndFor
    \State{Reduce column $n$ to a weight-1 column by adding a subset of column $0..n{-}1$ to it.  Such a linear combination exists since the matrix is full-rank} \Comment{${\leq}n$ $\cnotgate$ gates}
    \For{$i{=}1$ to $n$}  \Comment{$n$ $\cnotgate$ gates}
        \State{Find column $j$ such that it overlaps (shares a 1 position) with column $t$. Add column $t$ into it. Rename $t=j$}
    \EndFor
    \State{We now have a leftover matrix where, by construction, each column has weight 1. This is a SWAPping of qubits; implement this transformation using standard techniques.} \Comment{${\leq}3n$ $\cnotgate$ gates}
\end{algorithmic}
\end{algorithm}
}

Combining \alg{5} and \alg{6} we obtain the upper bound of $13n\,{=}\,6.5m+O(1) \in O(m)$ for the implementation of field multiplication by $x^{n+1}+1$ over odd $m\,{=}\,2n{+}1$.  We note that the lines 3-5 of \alg{6} may replace the lines 9-11 of \alg{5}, while deleting the line 15 of \alg{5}, which leads to the improved upper bound of $11n\,{=}\,5.5m+O(1)$.

One last question we need to address is the existence of a suitable irreducible polynomial.  The number of candidate polynomials for $m\,{=}\,2n{+}1$ is $1+3+5+...+n{-}3 + (n{-}1)/2$ (for odd $n$), roughly equal to $m^2/8$.  Since the probability that a random polynomial is irreducible is inversely linear in $m$, as $m$ grows, the candidate set becomes overwhelmingly large to contain an irreducible polynomial with a high degree of certainty.  We wrote code to explicitly confirm that a suitable polynomial can always be found for $m$ up to 10,000.  Given that the first classically unsolved instance of the discrete logarithm problem over an elliptic curve group relies on a Galois Field of size $m{=}163$ \cite{certicom} and the scaling of classical attacks is exponential, our calculations probably suffice to cover all relevant use cases for a quantum computer. 

The combination of \alg{1} (even $m$, trinomial irreducible polynomials), \alg{2} (even $m$, irreducible polynomial $x^m+x^{l_1}+x^{l_2}+\cdots+x^{l_k}+1$, where $n>l_1>l_2>\cdots>l_k$ and $l_1{-}l_k \in O(1)$), \alg{3} (odd $m$, trinomial irreducible polynomials), and \alg{5}-\alg{6} (odd $m$, $x^m + (x^{n-1} + x^{n-2} + ... + x^{n-l_1}) + (x^{l_2} + x^{l_2-1} + ... + 1)$, where $n{-}l_1 \,{>}\, l_2$), we now have the implementation of the multiplication by the constant $1{+}x^{\lceil m/2 \rceil}$ at the linear cost not exceeding $5.5m$.  These constructions are asymptotically optimal, since each output qubit is different from an input qubit in at least half the number of cases, implying a straightforward lower bound of $0.5m$.

We reported the sizes of quantum circuits for the multiplication by the constant $1{+}x^{\lceil m/2 \rceil}$ following the selection of irreducible polynomials and algorithms developed in this section in \fig{1}.  While the gate count in our circuits is upper bounded by $5.5m$, as shown in the previous paragraph, actual circuits never contain more than $4.157854m$ $\cnotgate$ gates, and the majority of circuits contain between $3.5m$ and $4.16m$ $\cnotgate$ gates.  This is a consequence of the upper bound being loose, and also due to considering up to five suitable irreducible polynomials of the specified type when there is a choice (this, in turn, implies that considering more candidates may improve the numbers further).  We highlight that the $\cnotgate$ gate count in the previous construction \cite{van2019space} visually scales faster than $m^{\log_2{3}}$, see \fig{1}, thus strongly suggesting that it was the bottleneck in the implementation of GF multiplication by the Karatsuba algorithm.  Our implementation suppresses the cost of the multiplication by the constant $1{+}x^{\lceil m/2 \rceil}$ to a linear function with a small leading coefficient, thus removing this bottleneck cost and allowing to claim $O\left(m^{\log_2{3}}\right)$ gate complexity for (ancilla-free) Karatsuba multiplication over $\text{GF}(2^m)$. 

\begin{figure}[th]

\centering
\includegraphics[width=0.99\textwidth]{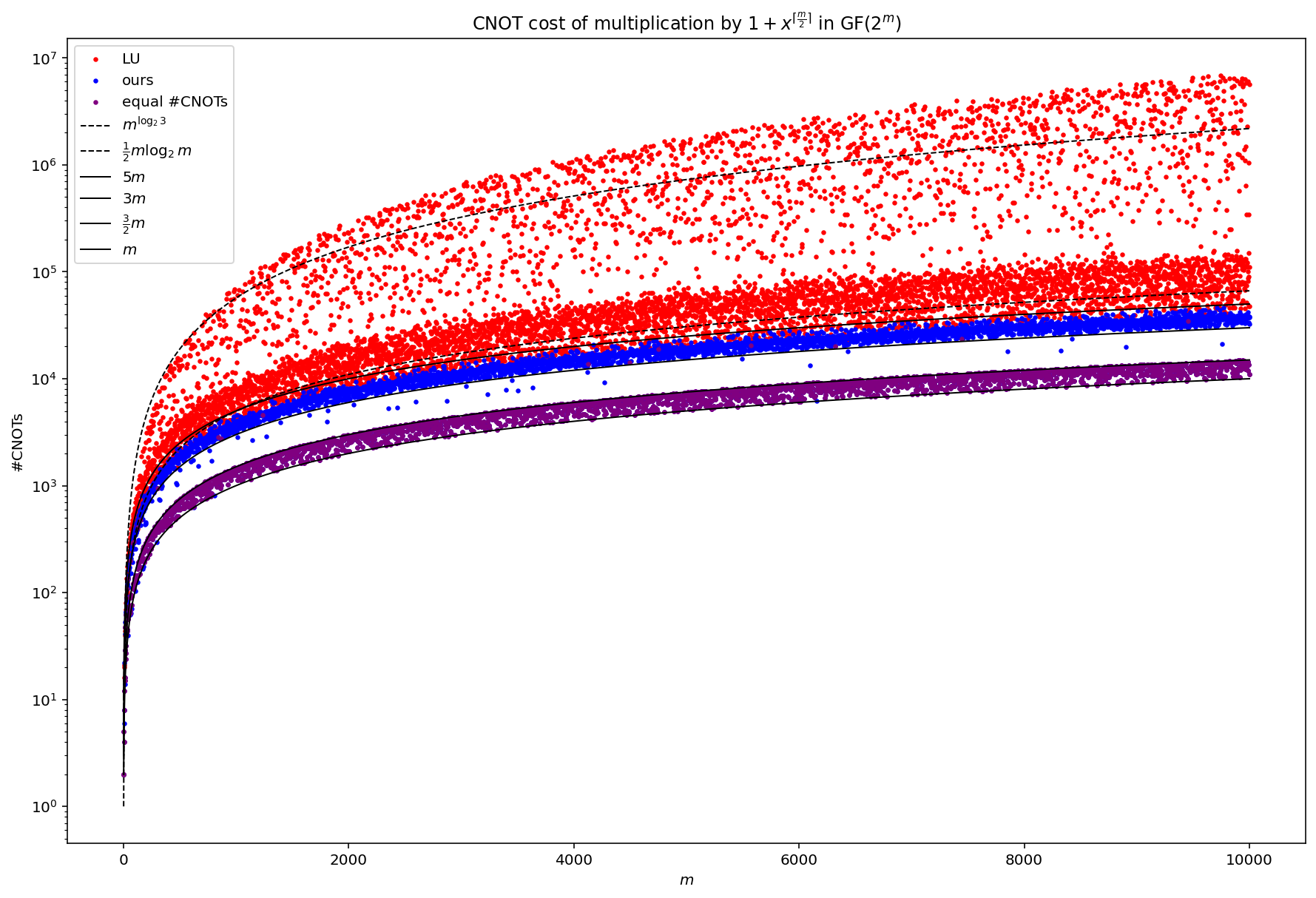}
\caption{Comparison of the gate count in the multiplication by the constant $1{+}x^{\lceil m/2 \rceil}$ over $\text{GF}(2^m)$ between state-of-the-art \cite{van2019space} LU implementation and ours.  Included are numeric experiments with $m$ up to 10,000.  The number of the $\cnotgate$ gates in our construction is upper bounded by $5.5m$, but in practice does not exceed $4.16m$.  Our optimization for larger $m$ consistently exceeds a factor of 100; for example, when $m{=}6159$, the LU decomposition requires 2,201,876 $\cnotgate$ gates, whereas ours takes 6,158 $\cnotgate$ gates, resulting in an improvement by a factor of 357.} \label{fig:1}
\end{figure}

\section{Cost of a $\sqrt{U}$ vs that of the $U$ itself}

Given a unitary transformation $U$ that one desires to implement as a quantum circuit, it can sometimes help to think of this $U$ as an evolution of a Hamiltonian. Then, both $U$ and its root $\sqrt{U}$ can be implemented by the algorithms simulating the evolution of the underlying Hamiltonian efficiently \cite{suzuki1991general, childs2018toward, low2019hamiltonian}. This approach, however, breaks when a root of $U$ is provably more complex than the implementation of $U$ itself.  Examples include $U\,{=}\,Id$, and its root given by the multiply-controlled $Z$ gate; indeed, the latter requires $\Omega(n)$ $\cnotgate$ gates \cite{shende2008cnot}. Here, we will look at transformations $U$ that are non-identity.  Another example can be given by the two-qubit controlled-$\pgate$ gate whose primary root, the controlled-$\tgate$ gate, cannot be implemented over two qubits as a Clifford+$\tgate$ circuit unless by approximation.  The proof is simple and relies on the determinant argument \cite{amy2013meet}. Here, we will require that both $\sqrt{U}$ and $U$ be implementable exactly, and the more expensive of them can even use ancillae and rely on arbitrary single- and two-qubit quantum gates. In addition, we will require that the library of gates that suffice to compute $\sqrt{U}$ and $U$ be as simple as possible.  This effectively implies that single-qubit circuits are ruled out, and one must consider circuits with entangling gates. $\czgate$ gate circuits do not work since $C^2\,{=}\,Id$ for every circuit $C$ composed with $\czgate$ gates.  The next simplest gate is the $\cnotgate$, and we thus focus on the linear reversible circuits.  Does there exist a linear reversible transformation $U$ whose root $V{:}\,\, V^2=U$ is both linear reversible and asymptotically more difficult to implement than $U$?  We answer this question positively while focusing on the circuit depth. 

\begin{lemma}
    For each even integer $n\,{=}\,2k$ there exist at least two families of linear reversible transformations that require circuit depth $\Omega(log(n))$ whose square can be implemented in constant depth, thus asymptotically faster. 
\end{lemma}
\begin{proof}
The two families of linear reversible transformations we consider, both defined for the number of qubits $n\,{=}\,2k$, are:
\begin{enumerate}
    \item[A.] Circulant invertible matrices with one $0$ per row.  Since circulant matrices are fully defined by their first row $r$, we will often focus on the row $r$ and its properties.  For each $n$, there are $n{-}2$ such matrices in this family. This is because the circulant matrices with the first row having its $0$ in the 0th or $\frac{n}{2}$th positions square to identity, and we agreed not to look at such matrices.
    \item[B.] Take $n=p(p{+}1)$, where $p\,{>}\,2$ is prime.  Consider the circulant matrix described by the row with ones in positions $0, p, 2p, 3p, ..., p^2{-}p$. The position numbering starts with zero.  The first row can be shifted by anywhere from $0$ to $n{-}1$ positions, thus leading to a family of $n$ different matrices; however, only $n{-}1$ remain after we remove those that square to the identity. 
\end{enumerate}

We first show that each matrix in the family A and B squares to a qubit permutation (in fact, a cyclic shift). To prove this, we need to establish that $(CS^t(r),r)$, the dot product of rows and their cyclic shifts (denoted $CS$) equals $1$ for exactly one value $t$, where $0 \,{\leq}\, t \,{<}\, n$. Indeed, all columns of a circulant matrix are also their rows, and thus a vector of products $(CS^t(r),r)$ describes the product of a circulant matrix by itself. Since circulant matrices form a group, it follows that a circulant matrix with one $1$ per row is a cyclic shift. 

For the family A, $(r,r)\,{=}\,1$ since the row $r$ has odd weight, and $(CS^t(r),r)\,{=}\,1$ for all $t{:} \,0 \,{<}\, t \,{<}\, n$ because the strings $CS^t(r)$ and $r$ are non-trivial shifts of each other, and thus their zero positions do not overlap.  This means the weight of the component-wise product is $n{-}2$. For family B, $(r,r)\,{=}\,1$ since the row $r$ has an odd weight $p$. For $t>0$:
\begin{itemize}
    \item  when $t$ is co-prime with $p$, the strings $CS^t(r)$ and $r$ do not overlap anywhere, and thus $(CS^t(r),r)=0$;
    \item when $t$ is a multiple of $p$, the string $r$ overlaps with the string $CS^t(r)$ in all places except $(t{-}p)^\text{th}$ and $(p^2-p+t)^\text{th}$---both positions are considered $\bmod\, n$. This means the weight of the overlap is even, because it is equal to $p{+}1{-}2 = p{-}1$, where $p$ is an odd prime.  This means $(CS^t(r),r)=0$.
\end{itemize}

Since family A contains transformations where an output qubit depends on $n{-}1$ input qubits, by the light cone argument, any quantum circuit computing this output qubit (including circuits with ancillae) must have depth at least $\log(n{-}1) \in \Omega(\log(n))$.  Similarly, family B contains an output qubit that depends on about $\sqrt{n}$ input qubits.  This means it can be computed by the circuits of depth at least $\Omega(\log(n))$.

Since the square of transformations in families A and B is a cyclic shift, and a cyclic shift (as well as any permutation) can be implemented in depth $2$ by the $\swapgate$ gates or depth $6$ by the $\cnotgate$ gates, all while using no ancillae, the square has only a constant depth complexity. 

We further note that circuit families $A$ and $B$ can be extended to include $n!-1$ matrices each by allowing arbitrary permutations of rows, except the one permutation resulting in a matrix that squares to the identity.  All proofs and arguments hold, but the sets are now much larger, containing $O(2^{n \log(n)})$ elements each.
\end{proof}

\section{Simultaneous Optimization of Linear Reversible Parts of GF Multiplication and Division}\label{sec:div}

Quantum algorithms relying on Galois Field arithmetic \cite{proos2003shor, jordan2024optimization} can often require the implementation of both field multiplication and division operations, even if the use of the more expensive division operation is minimized, e.g., through the use of projective coordinates \cite{maslov2009m2}.  This section is devoted to developing irreducible polynomials and relevant $\cnotgate$ gate circuitry to optimize the implementations of field multiplication and division operations simultaneously.

To this end, we develop a set of irreducible polynomials and $\cnotgate$ circuits that simultaneously reduce the cost of the multiplication by the constant $1{+}x^{\lceil m/2 \rceil}$ and field squaring operation.  The irreducible polynomial we consider is of the form $x^m{+}x{+}1 + (x^{2l_k}{+}x^{2l_{k-1}}{+}...{+}x^{2l_1})$, where $1 \leq l_k<l_{k-1}<...<l_1$ and $k$ is as small as possible.  Information theoretically, it suffices to take $k \in O(\log(m))$ to guarantee the existence of such an irreducible polynomial with a high probability for a given $m$.  \alg{2} and \alg{4} offer cost-$O(m \log(m))$ implementations of the multiplication by the constant $1{+}x^{\lceil m/2 \rceil}$ and \alg{7} offers a cost-$O(m \log(m))$ implementation of the field squaring operation. It works as follows:
\begin{itemize}
    \item[1:] The first $\lfloor m/2 \rfloor$ columns of the matrix $M$ have weight 1. We do not touch them.
    \item[3-8:] Add first $\lfloor m/2 \rfloor$ columns into the next $\lfloor m/2 \rfloor - l_1$ columns so as to reduce their weight to $1$.
    \item[9:] Reduce the weight of the rows corresponding to weight-1 columns considered in previous steps to 1. 
    \item[10:] Diagonalize the leftover $l_1{\times}l_1$ matrix obtained from $M$ by removing the columns and rows considered in the previous steps.
    \item[11:] The leftover matrix has weight-1 rows only, and as such, it is a SWAPping transformation. 
\end{itemize}

{\centering
\begin{algorithm}[H]
\caption{Implementation of squaring operation over $\text{GF}(2^m)$ with the irreducible polynomial $x^m{+}x{+}1 + (x^{2l_k}{+}x^{2l_{k-1}}{+}...{+}x^{2l_1})$, where $1 \leq l_k<l_{k-1}<...<l_1$} \label{alg:7}
\begin{algorithmic}[1]	
    \State{Input: matrix $M=\{m_{i,j}\}_{i,j=0..n}$ of size $m{\times}m$}
    \State{Output: circuit reducing matrix $M$ to the identity using $O(m \log(m))$ $\cnotgate$ gates}
    \For{$i{=}0$ to $\lfloor m/2 \rfloor$}  \Comment{$\leq ml_1/2$ $\cnotgate$ gates}
        \State{Add column $i$ to column $\lfloor m/2 \rfloor + 1$}
        \For{each $l_i \neq 0$}
            \State{Add column $l_i$ to column $\lfloor m/2 \rfloor + l_i$ if $m_{i,\lfloor m/2 \rfloor + l_j}=1$}
        \EndFor
    \EndFor
    \State{Use weight-1 columns $0..m-l_1$ to reduce the weight of the corresponding rows, $m-l_1+1..m$, to zero} \Comment{$\leq(m-l_1)l_1$ $\cnotgate$ gates}
    \State{Diagonalize the $l_1{\times}l_1$ matrix spanning last $l_1$ columns and corresponding non-zero rows} \Comment{$l_1^2$ $\cnotgate$ gates}
    \State{We now have a leftover matrix where, by construction, each column has weight 1. This is a SWAPping of qubits; implement this transformation using standard techniques.} \Comment{${\leq}3(m-1)$ $\cnotgate$ gates}
\end{algorithmic}
\end{algorithm}
}

To implement field inversion \cite{itoh1988fast, banegas2020concrete}, the squaring operation is used in batches of $1, 2, 4, ..., m/2$.  For each batch, \cite{banegas2020concrete} uses an LUP decomposition resulting in $m^2 + O(m)$ $\cnotgate$ gates. Clearly, iterating our circuit of complexity $O(m \log(m))$ will offer a better implementation for smaller batches, however, the longer batches may offer complexity as high as $O(m^2 \log(m))$, thus requiring the overall $\cnotgate$ gate complexity of $O(m^2 \log(m))$, same as in \cite{banegas2020concrete}.  The complexity may, however, be reduced to $O(m^2\log\log(m)/\log(m))$ by employing $O(m^2/ \log(m))$ \cite{patel2008optimal} algorithm to implement batches $m/2, m/4, ... , m/\log^2(m)$ using the asymptotically optimal technique \cite{patel2008optimal} and the remaining batches through the iteration of our circuit with $O(m \log(m))$ gates.  

\section{Reducing the Toffoli Count}

Several ancilla-free constructions of GF($2^m$) multiplication can be found in the literature, including constructions based on Karatsuba \cite{van2019space} and constructions based on CRT \cite{kim2024toffoli}. The CRT-based construction \cite{kim2024toffoli}, while having a very low $\toffoligate$ count $O(m \cdot 2^{\log_2^*{m}})$, has several parts that contribute the quadratic number $\cnotgate$ gates.  On the other hand, in Karatsuba constructions, while still offering the overall $O(m^2)$ $\cnotgate$ count, the quadratic term comes from only one term, and it is the multiplication by a constant, which we reduced to linear and linearithmic complexities above.  Thus, we focus on Karatsuba and, more generally, divide-and-conquer-based constructions.  Note that the techniques introduced in this paper can also be used to improve the gate count of the CRT-based construction.

\subsection{Toom-Cook Construction}
Our construction of quantum Toom-Cook for GF($2^m$) did not lead to improved gate counts due to the need to recurse on subproblems of degree $\frac{m}{k} + k - 1$ rather than $\frac{m}{k}$.  The details of our construction and the resulting gate counts are thus reported in Appendix \ref{appendix:ancilla_free_toom_cook}.

\subsection{Karatsuba-like Construction}
Inspired by \cite[Section 3]{kim2024toffoli}, we examined the Karatsuba-like multiplication formulas \cite{Weimerskirch2006GeneralizationsOT}. This approach divides the input polynomials into $k$ pieces (leading to a product with $2k{-}1$ pieces) and evaluates $p$ different products, each of size $\lceil \frac{m}{k} \rceil$, and then interpolates the $2k{-}1$ pieces of the product. A formula $C {=} R[(TA) \circ (TB)]$ represents the algorithm, where $T$ and $R$ are Boolean matrices of sizes $p \,{\times}\, k$ and $(2k{-}1)\,{\times}\, p$ respectively. The complexity of the multiplication algorithm depends on $p$, and finding $p$ and the corresponding matrix pair that minimizes the number of multiplications is an NP-hard problem \cite{better_circuits}. 

We turned the $(T, R)$ matrix pairs for $k {\in} \{2,3,4,5\}$ from \cite{better_circuits} and for $k {\in} \{6,7,8\}$ from \cite{searching_for_best_karatsuba} into ancilla-free quantum multiplication circuits.  We vary $k$ for each level of recursion to minimize the weighted cost, leading to lower gate counts, see \tab{multiplication}.

To reduce the $\cnotgate$ count, notice that the order we apply the rows of $T$ (columns of $R$) does not matter and that if two columns in $R$ are the same up to a shift, then they can be performed together. For example, the $R$ matrix for $k=2$ (i.e. ``pure'' Karatsuba) has the first column being $[1, 1, 0]^T$ meaning that $p_0$ should be multiplied by $1 + x^n$ and the last columns being $[0, 1, 1]^T$ meaning that $p_2$ should be multiplied by $x^n + x^{2n}$.  This means we can reorder the operations and group $p_0$ and $p_2$ together to apply the action of $1{+}x^n$ once, leading to a Karatsuba construction equivalent to \cite{van2019space}. A similar strategy works for any $(T_k, R_k)$ pair.

\subsection{Parity Controlled Toffoli Gate (PCTOF)} \label{sec:pctof}
The PCTOF gate was introduced in \cite{khattar2025verifiablequantumadvantageoptimized} as both a shorthand for a multi-target $\toffoligate$ whose two controls are the EXOR-sum of sets of qubits and to signify that this combined operation can be executed efficiently with lattice surgery.

The PCTOF representation is also useful as a way to find optimizations to the $\toffoligate$ count, as follows.  Given a circuit made only of $n$ PCTOF($c_1^{(i)}, c_2^{(i)}, t^{i}$), let $C_1 = \bigcup c_1^{(i)}$, $C_2 = \bigcup c_2^{(i)}$ and $T = \bigcup t^{(i)}$. If $C_1$, $C_2$, and $T$ are mutually disjoint, then we can replace the $n$ PCTOFs with $r$ PCTOF gates where $r$ is the rank of the binary matrix $n \times (|C_1|{+}|C_2|)$ where $(k, i|C_2|+j)$ is set if the kth PCTOF has qubit $i$ in its $c_1$ and qubit $j$ in its $c_2$.

Starting from a circuit made of $\cnotgate$ and $\toffoligate$ gates, we can replace the $\toffoligate$ gates with PCTOF gates and commute the $\cnotgate$ gates to the left while updating the PCTOF gates.  This leads to a circuit with two halves, where the first is made of only $\cnotgate$ gates and the second is made of only PCTOF gates. These operations can be done in reverse to restore the original circuit.

Applying the above transformation to the GF($2^m$) multiplication circuit, we get $U_1$ containing only qubits that represent the first input, $U_2$ containing only qubits representing the second input, and $T$ containing only qubits representing the target register. This allows us to use the rank reduction technique, leading to improvements that are as high as $5\%$ but can be as low as $0.2\%$ in the $\toffoligate$ count for odd $m$, but having no effect on $m$ that are powers of $2$.

\subsection{Addition Chains}
We can use FLT to compute the multiplicative inverse through exponentiation. For binary fields, this means $x^{-1} = x^{2^m {-} 2}$, the exponent $2^m {-} 2$ is a huge number, and direct computation through repeated squaring will use $\mathcal{O}(m)$ multiplications. Instead, the Itoh-Tsujii algorithm performs this operation using ${\leq}\, 2 \log_2{(m{-}1)}$ multiplications by exploiting the binary representation of $m{-}1$.

The simplest definition of an addition chain for a target value $t$ is a sequence of numbers $A(t) = \{a_0{=}1, a_1, \cdots, a_l {=} t\}$ where $\forall_{k > 0} \exists_{i,j < k} a_k = a_i + a_j$. Given such a sequence $A(2^{m} - 2)$, we get the number of squaring operations as the number of indices $k$ where $\exists_{i < k} a_k = 2 a_i$, all other $k$'s will need both a general multiplication and an ancilla register. The Itoh-Tsujii algorithm implicitly defines an addition chain that requires ${\leq}\, 2\log_2{(m{-}1)}$ general multiplications. We refer the reader to \cite{Bernstein2025, dimitrov2013another, TaguchiTakayasu, Taguchi2024} for deeper discussions on the addition chains.

In this work, we used the addition chains developed by Dimitrov and Järvinen \cite{dimitrov2013another}, which can reduce the number of multiplications needed for inversion by up to 22\%.

\section{Computational Aspects}

\subsection{Irreducible Polynomials}
A fundamental aspect of our paper is the reliance on irreducible polynomials of specific forms.  Contrary to the established use of trinomials or pentanomials \cite{seroussi1998table}, which are efficient for classical computations, employing them has not offered short ancilla-free quantum circuit implementations yet.  Instead, we developed irreducible polynomials that work well with our circuit constructions.  The number of binary polynomials of degree $m$ is $2^m$, and the number of irreducible ones is equal to \cite[A001037]{oeis} $\approx \frac{2^m}{m}$.  The exponential number of generic irreducible polynomials allowed us to find polynomials of every pattern we developed, because the number of candidates we tested is sufficiently greater than $m$. 

Specifically, we looked for patterns that produced between $m$ and $m^2$ candidates and tested them for irreducibility.  We performed fast prototyping of ideas in Python using the Galois library \cite{Hostetter_Galois_2020} and computed irreducible polynomials for $m$ up to 10,000 using C++ and the FLINT library \cite{flint}. 

To speed up our code, we did the following:

\begin{enumerate}
    \item To reduce the number of calls made to the costly irreducibility check, we tested divisibility with 100 small irreducible polynomials, similar to how (integer) primality tests check for small primes.
    \item We stopped the search after finding a small number (five) of irreducible polynomials for each pattern we tried.
    \item We used both multi-threading and multi-processing.
\end{enumerate}

\subsection{Gate Counting}
For small field sizes ($m \leq 64$), we built the full multiplication and division circuits for each irreducible polynomial we have and ran a full sweep of rank-based optimization of $\toffoligate$ count and local optimizations for $\cnotgate$ count. 

For bigger fields, we built the constant multiplication and squaring circuits for each irreducible polynomial and ran local optimizations \cite{maslov2008quantum} to reduce their sizes.  Given the tuple `(const\_mul, squaring, weight)` we can use dynamic programming to calculate the exact $\toffoligate$ count and an upper bound on the $\cnotgate$ count (that may be reducible using local optimization on the full circuit), where `const\_mul` and `squaring` are the sizes of the constant multiplication and squaring circuits respectively and weight is the weight of the irreducible polynomial.

To reduce the runtime, we ran the dynamic programming part only on irreducible polynomials whose cost tuple is not dominated by any other cost tuple.

\section{Numeric results}
For a most meaningful comparison, we need to focus on a specific circuit costing metric.  Due to the new developments in quantum error correction, including $\tgate$ gate cultivation at the cost comparable to that of the $\cnotgate$ \cite{gidney2024magic}, the ratio of costs of (local) $\toffoligate$ vs (local) $\cnotgate$ gate can be estimated to be about $1{:}10$ to $1{:}7$ \cite{huggins2025fluid}.  Thus, the cost of the $\cnotgate$ is no longer irrelevant and must be accounted for.  

\begin{table}[th]
    \centering
    \begin{tabular}{l|lrr|lrr}
    \toprule
     m & Irreducible polynomial & $\toffoligate$ & $\cnotgate$ & Our irreducible polynomial & Our $\toffoligate$ & Our $\cnotgate$ \\
    \midrule
    2 \cite{van2019space}  &  $x^2+x+1$  &  3  &  9  &  $x^2+x+1$  & 3 &  \textbf{7} \\
    4 \cite{van2019space}  &  $x^4+x+1$  &  9  &  44  &  $x^4+x+1$  & 9 &  \textbf{41} \\
    8 \cite{JangSrivastavaBaksi}  &  $x^8+x^4+x^3+x+1$  &  27  &  102  &  $x^8+x^4+x^3+x+1$  & 27 & 177\\
    16 \cite{JangSrivastavaBaksi}  &  $x^{16}+x^5+x^3+x+1$  &  81  &  655  &  $x^{{16}}+x^5+x^3+x+1$  & 81 &  \textbf{615} \\
    32 \cite{van2019space}  &  $x^{32}+x^7+x^3+x^2+1$  &  243  &  2,238  &  $x^{{32}}+x^8+x^5+x^2+1$  & 243 &  \textbf{2,004} \\
    64 \cite{van2019space}  &  $x^{64}+x^4+x^3+x+1$  &  729  &  6,896  &  $x^{{64}}+x^4+x^3+x^2+1$  & 729 &  \textbf{6,117} \\
    127 \cite{JangSrivastavaBaksi}  &  $x^{127}+x+1$  &  2,185  &  20,300  &  $x^{{127}}+x^{{63}}+1$  &  \textbf{2,179}  &  \textbf{18,336} \\
    128 \cite{van2019space}  &  $x^{128}+x^7+x^2+x+1$  &  2,187  &  21,272  &  $x^{{128}}+x^{{21}}+x^{{20}}+x^{{19}}+1$  & 2,187 &  \textbf{18,894} \\
    163 \cite{JangSrivastavaBaksi}  &  $x^{163}+x^7+x^6+x^3+1$  &  4,387  &  36,439  &  $x^{{163}}+x^7+x^6+x^3+1$  &  \textbf{3,605}  &  \textbf{35,070} \\
    169* \cite{van2019space}  &  $x^{169}+x^{34}+1$  &  4,627  &  41,313  &  $x^{{169}}+x^{{84}}+1$  &  \textbf{3,791}  &  \textbf{35,491} \\
    191* \cite{van2019space}  &  $x^{191}+x^9+1$  &  5,099  &  44,184  &  $x^{{191}}+x^{{14}}+x^6+x^2+1$  &  \textbf{4,365}  &  \textbf{39,252} \\
    233 \cite{JangSrivastavaBaksi}  &  $x^{233}+x^{74}+1$  &  6,323  &  60,453  &  $x^{{233}}+x^{{10}}+x^5+x+1$  &  \textbf{6,204}  &  \textbf{54,223} \\
    239* \cite{van2019space}  &  $x^{239}+x^{36}+1$  &  6,395  &  64,426  &  $x^{{239}}+x^5+x^4+x^3+x^2+x+1$  &  \textbf{6,312}  &  \textbf{54,494} \\
    256 \cite{van2019space}  &  $x^{256}+x^{10}+x^5+x^2+1$  &  6,561  &  64,706  &  $x^{{256}}+x^{{33}}+x^{{32}}+x^{{31}}+1$  & 6,561 &  \textbf{57,434} \\
    283 \cite{JangSrivastavaBaksi}  &  $x^{283}+x^{12}+x^7+x^5+1$  &  10,273  &  87,929  &  $x^{{283}}+x^{{12}}+x^7+x^5+1$  &  \textbf{8,712}  &  \textbf{79,994} \\
    295* \cite{van2019space}  &  $x^{295}+ x^{48} + 1$  &  11,281  &  101,966  &  $x^{{295}}+x^{{147}}+1$  &  \textbf{9,443}  &  \textbf{87,834} \\
    359* \cite{van2019space}  &  $x^{359}+x^{68}+1$  &  14,753  &  139,202  &  $x^{{359}}+x^{{11}}+x^7+x^4+1$  &  \textbf{12,630}  &  \textbf{111,904} \\
    409* \cite{van2019space}  &  $x^{{409}}+x^{{87}}+1$  &  17,101  &  172,327  &  $x^{{409}}+x^{{203}}+x^{{202}}+x^{{201}}+x^{{200}}+x^{{199}}+1$  &  \textbf{15,093}  &  \textbf{145,593} \\
    545* \cite{van2019space}  &  $x^{545}+x^{122}+1$  &  27,971  &  283,015  &  $x^{{545}}+x^8+x^7+x^6+x^2+x+1$  &  \textbf{24,103}  &  \textbf{238,577} \\
    551* \cite{van2019space}  &  $x^{551}+x^{135} + 1$  &  29,123  &  292,840  &  $x^{{551}}+x^9+x^4+x+1$  &  \textbf{25,800}  &  \textbf{231,556} \\
    571 \cite{JangSrivastavaBaksi}  &  $x^{571}+x^{10}+x^5+x^2+1$  &  31,171  &  267,771  &  $x^{{571}}+x^{{10}}+x^5+x^2+1$  &  \textbf{26,208}  &  \textbf{238,900} \\
    1,024 \cite{van2019space}  &  $x^{1024}+x^{19}+x^6+x+1$  &  59,049  &  591,942  &  $x^{{1024}}+x^{{39}}+x^{{37}}+x^{{36}}+1$  & 59,049 &  \textbf{525,140} \\
    \bottomrule
    \end{tabular}
    \caption{$\toffoligate$ and $\cnotgate$ gate counts in ancilla-free implementation of GF$(2^m)$ multiplication. The asterisk * indicates inputs $m$ for which the gate counts were produced by our reimplementation of the code in \cite{van2019space}; for all other numbers, our code independently verified the results from \cite{van2019space, JangSrivastavaBaksi}.}
    \label{tab:multiplication}
\end{table}
 

\begin{table}[t]
    \centering
    \begin{tabular}{l|lrrr|lrrr}
    \toprule
    m & Poly \cite{banegas2020concrete} & $\toffoligate$ \cite{banegas2020concrete} & $\cnotgate$ \cite{banegas2020concrete} & AR \cite{banegas2020concrete} & Our Poly & Our $\toffoligate$ & Our $\cnotgate$ & Our AR\\
    \midrule
    2*  &  $x^2+x+1$  &  3  &  11  &  0  &  $x^2+x+1$  & 3 &  \textbf{7}  & 0\\
    4*  &  $x^4+x+1$  &  45  &  252  &  2  &  $x^4+x+1$  & 45 &  \textbf{217}  & 2\\
    8  &  $x^8+x^4+x^3+x+1$  &  243  &  2,212  &  4  &  $x^8+x^4+x^3+x+1$  & 243 &  \textbf{1,791}  & 4\\
    16  &  $x^{{16}}+x^5+x^3+x+1$  &  1,053  &  10,814  &  6  &  $x^{{16}}+x^5+x^3+x+1$  &  \textbf{891}  &  \textbf{8,444}  & 6\\
    32*  &  $x^{{32}}+x^7+x^3+x^2+1$  &  4,131  &  44,674  &  8  &  $x^{{32}}+x^8+x^5+x^2+1$  &  \textbf{3,645}  &  \textbf{38,377}  & 8\\
    64*  &  $x^{{64}}+x^4+x^3+x+1$  &  15,309  &  169,272  &  10  &  $x^{{64}}+x^4+x^3+x^2+1$  &  \textbf{12,393}  &  \textbf{121,120}  &  \textbf{9} \\
    127  &  $x^{{127}}+x+1$  &  50,255  &  502,870  &  11  &  $x^{{127}}+x^{{63}}+1$  &  \textbf{41,401}  &  \textbf{382,060}  &  \textbf{10} \\
    128*  &  $x^{{128}}+x^7+x^2+x+1$  &  54,675  &  633,792  &  12  &  $x^{{128}}+x^7+x^2+x+1$  &  \textbf{45,927}  &  \textbf{538,810}  &  \textbf{11} \\
    163  &  $x^{{163}}+x^7+x^6+x^3+1$  &  83,353  &  906,170  &  9  &  $x^{{163}}+x^7+x^6+x^3+1$  &  \textbf{68,495}  &  \textbf{829,466}  & 9\\
    169*  &  $x^{{169}}+x^{{34}}+1$  &  87,913  &  872,279  &  9  &  $x^{{169}}+x^{{84}}+1$  &  \textbf{72,029}  &  \textbf{753,401}  & 9\\
    191*  &  $x^{{191}}+x^9+1$  &  127,475  &  1,172,714  &  12  &  $x^{{191}}+x^9+1$  &  \textbf{91,665}  &  \textbf{904,228}  &  \textbf{11} \\
    233  &  $x^{{233}}+x^{{74}}+1$  &  132,783  &  1,486,464  &  10  &  $x^{{233}}+x^{{74}}+1$  &  \textbf{130,284}  &  \textbf{1,286,813}  & 10\\
    239*  &  $x^{{239}}+x^{{36}}+1$  &  159,875  &  1,769,182  &  12  &  $x^{{239}}+x^{{36}}+1$  &  \textbf{157,800}  &  \textbf{1,552,052}  & 12\\
    256*  &  $x^{{256}}+x^{{10}}+x^5+x^2+1$  &  190,269  &  2,186,950  &  14  &  $x^{{256}}+x^{{10}}+x^5+x^2+1$  &  \textbf{137,781}  &  \textbf{1,621,158}  &  \textbf{11} \\
    283  &  $x^{{283}}+x^{{12}}+x^7+x^5+1$  &  236,279  &  2,708,404  &  11  &  $x^{{283}}+x^{{12}}+x^7+x^5+1$  &  \textbf{200,376}  &  \textbf{2,334,404}  & 11\\
    295*  &  $x^{{295}}+x^{{48}}+1$  &  259,463  &  2,628,060  &  11  &  $x^{{295}}+x^{{147}}+1$  &  \textbf{198,303}  &  \textbf{2,026,538}  & 11\\
    359*  &  $x^{{359}}+x^{{68}}+1$  &  368,825  &  3,848,304  &  12  &  $x^{{359}}+x^{{68}}+1$  &  \textbf{315,750}  &  \textbf{3,241,226}  & 12\\
    409*  &  $x^{{409}}+x^{{87}}+1$  &  393,323  &  4,304,131  &  11  &  $x^{{409}}+x^{{87}}+1$  &  \textbf{316,953}  &  \textbf{3,646,883}  & 11\\
    545*  &  $x^{{545}}+x^{{122}}+1$  &  587,391  &  6,910,325  &  10  &  $x^{{545}}+x^{{122}}+1$  &  \textbf{506,163}  &  \textbf{6,197,937}  & 10\\
    551*  &  $x^{{551}}+x^{{135}}+1$  &  728,075  &  8,066,560  &  12  &  $x^{{551}}+x^{{135}}+1$  &  \textbf{645,000}  &  \textbf{6,737,686}  & 12\\
    571  &  $x^{{571}}+x^{{10}}+x^5+x^2+1$  &  814,617  &  10,941,536  &  13  &  $x^{{571}}+x^{{10}}+x^5+x^2+1$  &  \textbf{653,525}  &  \textbf{8,941,016}  & 13\\
    1024*  &  $x^{{1024}}+x^{{19}}+x^6+x+1$  &  2,184,813  &  28,318,894  &  18  &  $x^{{1024}}+x^{{39}}+x^{{37}}+x^{{36}}+1$  &  \textbf{1,594,323}  &  \textbf{22,961,736}  &  \textbf{14} \\
    \bottomrule
    \end{tabular}
    \caption{$\toffoligate$, $\cnotgate$ gate counts and number of Ancilla Registers (AR) in Fermat's Little Theorem-based implementation of the GF$(2^m)$ division.  The asterisk * indicates inputs $m$ for which the gate counts were produced by our reimplementation of the code in \cite{banegas2020concrete}; for all other numbers, our code independently verified the original results \cite{banegas2020concrete}.}
    \label{tab:division}
\end{table}

In \tab{multiplication}, we compare our results for multiplication with \cite{van2019space} and \cite{JangSrivastavaBaksi}. In \tab{division}, we compare our results for GF division to those in \cite{banegas2020concrete}, where we were able to independently verify the constructions from each.  In our paper, we used the numbers of inputs $m$ employed by the previous authors, and also added cryptographically relevant numbers from \cite{certicom, dimitrov2013another} as well as numbers 169, 295, 545, and 551 where our linear circuit construction for the multiplication by the constant $1{+}x^{\lceil m/2 \rceil}$ turned out to be substantially more efficient than the one from \cite{van2019space}.  For field sizes where none of them report a result, we used our reproduction of their construction, where we use the multiplication circuit from \cite{van2019space} and the FLT-construction from \cite{banegas2020concrete} with the irreducible polynomials from \cite{seroussi1998table}.
 



\section{Discussion}

We tackled a well-studied problem of implementing GF arithmetic as quantum/reversible circuits, and yet found both asymptotic and practical optimizations.  We showed that GF multiplication can be implemented using $O\left(m^{\log_2{3}}\right)$ gates with no ancilla, improving the best previously known asymptotic gate count of $O(m^2)$ \cite{van2019space, putranto2023depth, kim2024toffoli}.  Numeric results, \tab{multiplication}, highlight optimization of the $\cnotgate$ gate count by up to 20.9\% and $\toffoligate$ gate count by up to 18.06\% for cryptographically relevant numbers $m$; the $\cnotgate$ gate reduction is expected to grow with $m$ due to improvement in asymptotics.  Our optimization was achieved by reducing the complexity of the bottleneck part performing in-place multiplication by $1+x^{\lceil m/2 \rceil}$ from $O(m^2)$ to $O(m)$ gates.  The leading constant in front of our linear cost construction is upper bounded by $4.16$, leading to the gate count reduction by as much as a factor of 357 for $m$ up to 10,000.  Our construction is useful not only in quantum computation, but it implies that $\text{GF}(2^m)$ multiplication of two $m$-bit registers can be accomplished by a classical program with complexity $O\left(m^{\log_2{3}}\right)$ and access to a total of only $3m$ bits of memory.

We note that the focus on other circuit cost metrics may result in different constructions.  For instance, $\text{GF}(2^m)$ multiplication can be implemented in $\tgate$ depth $1$ by Mastrovito multiplier \cite{mastrovito1988vlsi, maslov2009m2}.  Indeed, the circuit consists of three parts: $\toffoligate$, $\cnotgate$, and $\toffoligate$, where Toffolis have controls on the top two multiregisters $a$ and $b$, and $\cnotgate$ gates operate on the bottom third multiregister $c$.  Such a circuit can be rewritten as $\hgate^c.\cczgate^{abc}.\hgate^c.\cnotgate^{c\downarrow}.\hgate^c.\cczgate^{abc}.\hgate^c = \hgate^c.\cczgate^{abc}.\cnotgate^{c\uparrow}.\cczgate^{abc}.\hgate^c$, where $\cnotgate^{c\uparrow}$ is obtained from $\cnotgate^{c\downarrow}$ by replacing targets with controls and controls with targets in the original circuit.  The part $\cczgate^{abc}.\cnotgate^{c\uparrow}.\cczgate^{abc}$ is a $\cnotgate{+}\tgate$ circuit, and thus can be written in $\tgate$ depth $1$ \cite{amy2013meet}.  We do not consider such a circuit practical since it relies on $O(m^2)$ ancillae and $O(m^2)$ gates.

Our construction of the GF division improves the gate count complexity from $O(m^2\log(m))$ to $O(m^2\log\log(m)/\log(m))$.  This was achieved through selecting an irreducible polynomial that allows simultaneous implementation of two technical computing primitives, in-place multiplication by the constant $1{+}x^{\lceil m/2 \rceil}$ and in-place field squaring operation, each at the cost of $O(m\log(m))$ gates.  In practice, the $\cnotgate$ and $\toffoligate$ gate counts can be reduced by 28\% (e.g. $m{=}64$) and 28\% (e.g. $m{=}191$) respectively, while occasionally reducing the number of ancillary registers by up to 22\% (e.g. $m{=}1024$), see \tab{division}.

Finally, we reported an example of a unitary $U$ and its root $\sqrt{U}$ such that both $\sqrt{U}$ and $U$ can be implemented using the $\cnotgate$ gates alone, and yet $\sqrt{U}$ requires asymptotically higher depth than $U$ even if ancillae are allowed.  This offers an example of a scenario where a root is more complex than the unitary $U$ itself.

\section{Data Availability}
Circuits created for \tab{multiplication} and \tab{division} are 
\href{https://doi.org/10.5281/zenodo.19193656}{available on Zenodo}. 

\section*{Acknowledgments}
We would like to thank Vassil Dimitrov, University of Calgary and Lemurian Labs, and Kimmo Jarvinen, Xiphera Inc., for teaching us the addition chains and for the excellent discussions about the Itoh-Tsujii algorithm.

\bibliographystyle{plain}
\bibliography{biblio}

@inproceedings{mastrovito1988vlsi,
  title={{VLSI} designs for multiplication over finite fields {GF}$(2^m)$},
  author={Mastrovito, Edoardo D.},
  booktitle={International Conference on Applied Algebra, Algebraic Algorithms, and Error-Correcting Codes},
  pages={297--309},
  year={1988},
  organization={Springer}
}

@article{maslov2009m2,
  title={An ${O}(m^2)$-depth quantum algorithm for the elliptic curve discrete logarithm problem over {GF}$(2^m)$},
  author={Maslov, Dmitri and Mathew, Jimson and Cheung, Donny and Pradhan, Dhiraj K.},
  journal={Quantum Information \& Computation},
  volume={9},
  number={7},
  pages={610--621},
  year={2009},
  publisher={Rinton Press, Incorporated Paramus, NJ}
}

@article{amy2013meet,
  title={A meet-in-the-middle algorithm for fast synthesis of depth-optimal quantum circuits},
  author={Amy, Matthew and Maslov, Dmitri and Mosca, Michele and Roetteler, Martin},
  journal={IEEE Transactions on Computer-Aided Design of Integrated Circuits and Systems},
  volume={32},
  number={6},
  pages={818--830},
  year={2013},
  publisher={IEEE}
}

@article{patel2008optimal,
  title={Optimal synthesis of linear reversible circuits},
  author={Patel, Ketan N. and Markov, Igor L. and Hayes, John P.},
  journal={Quantum Information \& Computation},
  volume={8},
  number={3},
  pages={282--294},
  year={2008},
  publisher={Citeseer}
}

@article{itoh1989structure,
  title={Structure of parallel multipliers for a class of fields {GF}$(2^m)$},
  author={Itoh, Toshiya and Tsujii, Shigeo},
  journal={Information and Computation},
  volume={83},
  number={1},
  pages={21--40},
  year={1989},
  publisher={Elsevier}
}

@article{takagi2001fast,
  title={A fast algorithm for multiplicative inversion in {GF}$(2^m)$ using normal basis},
  author={Takagi, Naofumi and Yoshiki, Jun-ichi and Takagi, Kazuyoshi},
  journal={IEEE Transactions on Computers},
  volume={50},
  number={5},
  pages={394--398},
  year={2001},
  publisher={IEEE}
}

@article{amento2012efficient,
  title={Efficient quantum circuits for binary elliptic curve arithmetic: reducing {T}-gate complexity},
  author={Amento, Brittanney and Steinwandt, Rainer and Roetteler, Martin},
  journal={arXiv preprint arXiv:1209.6348},
  year={2012}
}

@inproceedings{karatsuba1962multiplication,
  title={Multiplication of many-digital numbers by automatic computers},
  author={Karatsuba, Anatolii Alekseevich and Ofman, Yu. P.},
  booktitle={Doklady Akademii Nauk},
  volume={145},
  number={2},
  pages={293--294},
  year={1962},
  organization={Russian Academy of Sciences}
}

@article{gidney2019asymptotically,
  title={Asymptotically efficient quantum {K}aratsuba multiplication},
  author={Gidney, Craig},
  journal={arXiv preprint arXiv:1904.07356},
  year={2019}
}

@article{van2019space,
  title={Space-efficient quantum multiplication of polynomials for binary finite fields with sub-quadratic {T}offoli gate count},
  author={Van Hoof, Iggy},
  journal={arXiv preprint arXiv:1910.02849},
  year={2019}
}

@article{chebolu2011counting,
  title={Counting irreducible polynomials over finite fields using the inclusion-exclusion principle},
  author={Chebolu, Sunil K. and Min{\'a}{\v{c}}, J{\'a}n},
  journal={Mathematics Magazine},
  volume={84},
  number={5},
  pages={369--371},
  year={2011},
  publisher={Taylor \& Francis}
}

@article{shende2008cnot,
  title={On the {CNOT}-cost of {T}OFFOLI gates},
  author={Shende, Vivek V. and Markov, Igor L.},
  journal={arXiv preprint arXiv:0803.2316},
  year={2008}
}

@article{low2019hamiltonian,
  title={Hamiltonian simulation by qubitization},
  author={Low, Guang Hao and Chuang, Isaac L.},
  journal={Quantum},
  volume={3},
  pages={163},
  year={2019},
  publisher={Verein zur F{\"o}rderung des Open Access Publizierens in den Quantenwissenschaften}
}

@article{childs2018toward,
  title={Toward the first quantum simulation with quantum speedup},
  author={Childs, Andrew M. and Maslov, Dmitri and Nam, Yunseong and Ross, Neil J. and Su, Yuan},
  journal={Proceedings of the National Academy of Sciences},
  volume={115},
  number={38},
  pages={9456--9461},
  year={2018},
  publisher={National Academy of Sciences}
}

@article{suzuki1991general,
  title={General theory of fractal path integrals with applications to many-body theories and statistical physics},
  author={Suzuki, Masuo},
  journal={Journal of Mathematical Physics},
  volume={32},
  number={2},
  pages={400--407},
  year={1991},
  publisher={American Institute of Physics}
}

@article{proos2003shor,
  title={Shor's discrete logarithm quantum algorithm for elliptic curves},
  author={Proos, John and Zalka, Christof},
  journal={arXiv preprint quant-ph/0301141},
  year={2003}
}

@article{itoh1988fast,
  title={A fast algorithm for computing multiplicative inverses in {GF}$(2^m)$ using normal bases},
  author={Itoh, Toshiya and Tsujii, Shigeo},
  journal={Information and Computation},
  volume={78},
  number={3},
  pages={171--177},
  year={1988},
  publisher={Elsevier}
}

@article{banegas2020concrete,
  title={Concrete quantum cryptanalysis of binary elliptic curves},
  author={Banegas, Gustavo and Bernstein, Daniel J and Van Hoof, Iggy and Lange, Tanja},
  journal={IACR Transactions on Cryptographic Hardware and Embedded Systems},
  pages={451--472},
  year={2021}
}

@article{certicom,
    title = {The {C}erticom {ECC} Challenge},
    author = {Certicom},
    journal = {https://www.certicom.com/content/certicom/en/the-certicom-ecc-challenge.html, last accessed 5/9/2025}
}

@article{jordan2024optimization,
  title={Optimization by decoded quantum interferometry},
  author={Jordan, Stephen P and Shutty, Noah and Wootters, Mary and Zalcman, Adam and Schmidhuber, Alexander and King, Robbie and Isakov, Sergei V and Khattar, Tanuj and Babbush, Ryan},
  journal={Nature},
  volume={646},
  number={8086},
  pages={831--836},
  year={2025},
  publisher={Nature Publishing Group UK London}
}

@article{kahanamoku2022classically,
  title={Classically verifiable quantum advantage from a computational {B}ell test},
  author={Kahanamoku-Meyer, Gregory D. and Choi, Soonwon and Vazirani, Umesh V. and Yao, Norman Y.},
  journal={Nature Physics},
  volume={18},
  number={8},
  pages={918--924},
  year={2022},
  publisher={Nature Publishing Group UK London}
}

@article{gidney2024magic,
  title={Magic state cultivation: growing {T} states as cheap as {CNOT} gates},
  author={Gidney, Craig and Shutty, Noah and Jones, Cody},
  journal={arXiv preprint arXiv:2409.17595},
  year={2024}
}

@misc{oeis,
    Author = {{OEIS Foundation Inc.}},
    Note = {Published electronically at \url{http://oeis.org}},
    Title = {The {O}n-{L}ine {E}ncyclopedia of {I}nteger {S}equences}}

@software{Hostetter_Galois_2020,
   title = {{Galois: A performant NumPy extension for Galois Fields}},
   author = {Hostetter, Matt},
   month = {11},
   year = {2020},
   url = {https://github.com/mhostetter/galois},
}

@manual{flint,
  key = {{FLINT}},
  author = {The {FLINT} team},
  title = {{FLINT}: {F}ast {L}ibrary for {N}umber {T}heory},
  year = {2025},
  note = {Version 3.2.1, \url{https://flintlib.org}}
}

@book{seroussi1998table,
  title={Table of low-weight binary irreducible polynomials},
  author={Seroussi, Gadiel},
  year={1998},
  publisher={Hewlett-Packard Laboratories Palo Alto, California}
}

@article{kepley2015quantum,
  title={Quantum circuits for $\mathbb{F}_2^n$-multiplication with subquadratic gate count},
  author={Kepley, Shane and Steinwandt, Rainer},
  journal={Quantum Information Processing},
  volume={14},
  pages={2373--2386},
  year={2015},
  publisher={Springer}
}

@article{huggins2025fluid,
  title={The {F}Luid {A}llocation of {S}urface code {Q}ubits ({FLASQ}) cost model for early fault-tolerant quantum algorithms},
  author={Huggins, William J. and Khattar, Tanuj and Xu, Amanda and Harrigan, Matthew and Kang, Christopher and Low, Guang Hao and Fowler, Austin and Rubin, Nicholas C. and Babbush, Ryan},
  journal={arXiv preprint arXiv:2511.08508},
  year={2025}
}

@article{kim2024toffoli,
  title={Toffoli gate count optimized space-efficient quantum circuit for binary field multiplication},
  author={Kim, Sunyeop and Kim, Insung and Kim, Seonggyeom and Hong, Seokhie},
  journal={Quantum Information Processing},
  volume={23},
  number={10},
  pages={330},
  year={2024},
  publisher={Springer}
}

@article{putranto2023depth,
  title={Depth-optimization of quantum cryptanalysis on binary elliptic curves},
  author={Putranto, Dedy Septono Catur and Wardhani, Rini Wisnu and Larasati, Harashta Tatimma and Ji, Janghyun and Kim, Howon},
  journal={IEEE Access},
  volume={11},
  pages={45083--45097},
  year={2023},
  publisher={IEEE}
}

@article{maslov2008quantum,
  title={Quantum circuit simplification and level compaction},
  author={Maslov, Dmitri and Dueck, Gerhard W. and Miller, D. Michael and Negrevergne, Camille},
  journal={IEEE Transactions on Computer-Aided Design of Integrated Circuits and Systems},
  volume={27},
  number={3},
  pages={436--444},
  year={2008},
  publisher={IEEE}
}

@InProceedings{zimmermann_faster_multiplication,
author="Brent, Richard P.
and Gaudry, Pierrick
and Thom{\'e}, Emmanuel
and Zimmermann, Paul",
editor="van der Poorten, Alfred J.
and Stein, Andreas",
title="Faster Multiplication in {GF}(2)[x]",
booktitle="Algorithmic Number Theory",
year="2008",
publisher="Springer Berlin Heidelberg",
address="Berlin, Heidelberg",
pages="153--166",
isbn="978-3-540-79456-1"
}

@InProceedings{towards_optimal_toom_cook,
author="Bodrato, Marco",
editor="Carlet, Claude
and Sunar, Berk",
title="Towards Optimal {T}oom-{C}ook Multiplication for Univariate and Multivariate Polynomials in Characteristic 2 and 0",
booktitle="Arithmetic of Finite Fields",
year="2007",
publisher="Springer Berlin Heidelberg",
address="Berlin, Heidelberg",
pages="116--133",
isbn="978-3-540-73074-3"
}

@article{Larasati2021QuantumCD,
  title={Quantum Circuit Design of {T}oom 3-Way Multiplication},
  author={Harashta Tatimma Larasati and Asep Muhamad Awaludin and Janghyun Ji and Howon Kim},
  journal={Applied Sciences},
  year={2021},
  volume={11},
  pages={3752},
  url={https://api.semanticscholar.org/CorpusID:234863981}
}

@article{Dutta_2018,
   title={Quantum circuits for {T}oom-{C}ook multiplication},
   volume={98},
   ISSN={2469-9934},
   url={http://dx.doi.org/10.1103/PhysRevA.98.012311},
   DOI={10.1103/physreva.98.012311},
   number={1},
   journal={Physical Review A},
   publisher={American Physical Society (APS)},
   author={Dutta, Srijit and Bhattacharjee, Debjyoti and Chattopadhyay, Anupam},
   year={2018},
   month=jul }

@article{kahanamokumeyer2024fastquantumintegermultiplication,
      title={Fast quantum integer multiplication with zero ancillas}, 
      author={Gregory D. Kahanamoku-Meyer and Norman Y. Yao},
      year={2024},
      journal={arXiv preprint arXiv:2403.18006},
      eprint={2403.18006},
      archivePrefix={arXiv},
      primaryClass={quant-ph},
      url={https://arxiv.org/abs/2403.18006}, 
}

@ARTICLE{better_circuits,
  author={Find, Magnus Gaudal and Peralta, René},
  journal={IEEE Transactions on Computers}, 
  title={Better Circuits for Binary Polynomial Multiplication}, 
  year={2019},
  volume={68},
  number={4},
  pages={624-630},
  doi={10.1109/TC.2018.2874662}}

@InProceedings{searching_for_best_karatsuba,
author="{\c{C}}al{\i}k, {\c{C}}a{\u{g}}da{\c{s}}
and Dworkin, Morris
and Dykas, Nathan
and Peralta, Rene",
editor="Kotsireas, Ilias
and Pardalos, Panos
and Parsopoulos, Konstantinos E.
and Souravlias, Dimitris
and Tsokas, Arsenis",
title="Searching for Best {K}aratsuba Recurrences",
booktitle="Analysis of Experimental Algorithms",
year="2019",
publisher="Springer International Publishing",
address="Cham",
pages="332--342",
isbn="978-3-030-34029-2"
}

@article{Weimerskirch2006GeneralizationsOT,
  title={Generalizations of the {K}aratsuba Algorithm for Efficient Implementations},
  author={Andr{\'e} Weimerskirch and Christof Paar},
  journal={IACR Cryptol. ePrint Arch.},
  year={2006},
  volume={2006},
  pages={224},
  url={https://api.semanticscholar.org/CorpusID:16301209}
}

@article{khattar2025verifiablequantumadvantageoptimized,
      title={Verifiable Quantum Advantage via Optimized {DQI} Circuits}, 
      author={Tanuj Khattar and Noah Shutty and Craig Gidney and Adam Zalcman and Noureldin Yosri and Dmitri Maslov and Ryan Babbush and Stephen P. Jordan},
      year={2025},
      journal={arXiv preprint arXiv:2510.10967},
      eprint={2510.10967},
      archivePrefix={arXiv},
      primaryClass={quant-ph},
      url={https://arxiv.org/abs/2510.10967}, 
}

@article{vandaele2025quantumbinaryfieldmultiplication,
      title={Quantum binary field multiplication with subquadratic {T}offoli gate count and low space-time cost}, 
      author={Vivien Vandaele},
      year={2025},
      journal={arXiv preprint arXiv:2501.16136},
      eprint={2501.16136},
      archivePrefix={arXiv},
      primaryClass={quant-ph},
      url={https://arxiv.org/abs/2501.16136}, 
}

@article{roetteler2013discretelogarithms,
      title={A quantum circuit to find discrete logarithms on ordinary binary elliptic curves in depth $\mathcal{O}(log^2 n$)}, 
      author={Martin Roetteler and Rainer Steinwandt},
      year={2013},
      journal={arXiv preprint arXiv:1306.1161},
      eprint={1306.1161},
      archivePrefix={arXiv},
      primaryClass={quant-ph},
      url={https://arxiv.org/abs/1306.1161}, 
}

@misc{BY_GCD,
      author = {Daniel J.  Bernstein and Bo-Yin Yang},
      title = {Fast constant-time gcd computation and modular inversion},
      howpublished = {Cryptology {ePrint} Archive, Paper 2019/266},
      year = {2019},
      url = {https://eprint.iacr.org/2019/266}
}

@misc{JangSrivastavaBaksi,
      author = {Kyungbae Jang and Vikas Srivastava and Anubhab Baksi and Santanu Sarkar and Hwajeong Seo},
      title = {New Quantum Cryptanalysis of Binary Elliptic Curves (Extended Version)},
      howpublished = {Cryptology {ePrint} Archive, Paper 2025/017},
      year = {2025},
      doi = {10.46586/tches.v2025.i2.781-804},
      url = {https://eprint.iacr.org/2025/017}
}

@misc{TaguchiTakayasu,
      author = {Ren Taguchi and Atsushi Takayasu},
      title = {On the Untapped Potential of the Quantum {FLT}-based Inversion},
      howpublished = {Cryptology {ePrint} Archive, Paper 2024/228},
      year = {2024},
      doi = {10.1007/978-3-031-54773-7_4},
      url = {https://eprint.iacr.org/2024/228}
}

@inproceedings{dimitrov2013another,
  title={Another look at inversions over binary fields},
  author={Dimitrov, Vassil and J{\"a}rvinen, Kimmo},
  booktitle={2013 IEEE 21st Symposium on Computer Arithmetic},
  pages={211--218},
  year={2013},
  organization={IEEE}
}

@Article{Bernstein2025,
author={Bernstein, Daniel J.
and Cottaar, Jolijn
and Lange, Tanja},
title={Searching for differential addition chains},
journal={Research in Number Theory},
year={2025},
month={Mar},
day={27},
volume={11},
number={2},
pages={45},
issn={2363-9555},
doi={10.1007/s40993-024-00604-8},
url={https://doi.org/10.1007/s40993-024-00604-8}
}

@Article{Taguchi2024,
author={Taguchi, Ren
and Takayasu, Atsushi},
title={Concrete quantum cryptanalysis of binary elliptic curves via addition chain},
journal={Quantum Information Processing},
year={2024},
month={Mar},
day={25},
volume={23},
number={4},
pages={122},
issn={1573-1332},
doi={10.1007/s11128-024-04323-y},
url={https://doi.org/10.1007/s11128-024-04323-y}
}

\appendix

\section{Ancilla-free Multiplication Circuit for Toom-Cook Algorithm over $GF(2^m)$} \label{appendix:ancilla_free_toom_cook}

\subsection{Background}
Few previous studies \cite{Larasati2021QuantumCD, Dutta_2018, kahanamokumeyer2024fastquantumintegermultiplication} utilized the Toom-Cook algorithm for integer multiplication. The first two papers used a polynomial number of ancillae, while the third requires no ancilla but does not explicitly offer circuit construction. 

Polynomials over GF$(2)$ have exactly three possible evaluation points $\{0, 1, \infty\}$ that are enough for Karatsuba, giving the evaluation points $f_0 g_0$, $(f_0+f_1)(g_0 + g_1)$, and $f_1 g_1$, respectively. For Toom-Cook (Toom-$k$), we need $2k{-}1$ points, since we divide the input polynomials $f(x)$ and $g(x)$ into $k$ parts $f(x) = \sum_{i=0}^{k-1} f_i x^{i \frac{m}{k}}$ whose product $h(x)$ has $2k{-}1$ parts $h(x) = \sum_{i=0}^{2k-2} w_i x^{i \frac{m}{k}}$.  In \cite{zimmermann_faster_multiplication}, Zimmermann and Quercia discovered that they can use any polynomial substitutions for $x^{m/k}$ and for Toom-3, they evaluated the polynomial at $\{0, 1, \infty, x, x^{-1}\}$.  In \cite{towards_optimal_toom_cook}, Bodrato generalized the evaluation points to be any rational function $\frac{A(x)}{B(x)}$ corresponding to $f(A(x), B(x)) \to \sum_{i=0}^{k-1} f_i A(x)^i B(x)^{k-i-1}$ and showed that the evaluation points $\{(0, 1), (1, 1), (1, 0), (x, 1), (x+1, 1)\} \to \{0, 1, \infty, x, x+1\}$ are classically optimal.

A downside of using polynomial evaluation points is that we do not recurse on a problem of size $\frac{m}{k}$ but on a problem of size $\frac{m}{k} + t(k-1)$ where $t$ is the maximum of the degrees of $A(x)$ and $B(x)$. For example, for points $(1, x)$ and $(x, 1)$ we recurse on problems of size $\frac{m}{k} + k-1$.

\subsection{Construction}
Given a list of evaluation points $\frac{A_0(x)}{B_0(x)}, \frac{A_1(x)}{B_1(x)}, \cdots, \frac{A_{2k-1}(x)}{B_{2k-1}(x)}$ leading to products $p_0, p_1, \cdots, p_{2k-1}$, we can compute the coefficients $w_i$ using the Vandermonde matrix $V$ as $P = VW \implies W = V^{-1}P$

\begin{equation}
    V = \begin{bmatrix}
        B_0^{2k-1}& \cdots& A_0^{2k-1}\\
        \vdots{} & \ddots{} & \vdots{}\\
        B_{2k-1}^{2k-1} & \cdots & A_{2k-1}^{2k-1}
    \end{bmatrix}
\end{equation}

In order to avoid using ancillae, we skip the computation of the $w_i$ values and compute the product directly as $e V^{-1} P$ where $e = [1, x^{m/k}, \cdots, x^{\frac{2k-2}{k}m}]$. This is similar to Appendix A of \cite{kahanamokumeyer2024fastquantumintegermultiplication} and leads to the Karatsuba algorithm in the same way and reproduces the construction in \cite{van2019space} for $GF(2)[x]$.

{\centering
\begin{algorithm}[H]
\caption{A general framework for constructing Toom-$k$ circuits for GF(2)[x] multiplication} \label{alg:toom_cook_framework}
\begin{algorithmic}[1]	
    \State{Input: Integer k and a list of fractional evaluation points $\frac{A_0(x)}{B_0(x)}, \frac{A_1(x)}{B_1(x)}, \cdots, \frac{A_{2k-1}(x)}{B_{2k-1}(x)}$. Quantum registers $f$, $g$, and $h$}
    \State{Output: a circuit applying $h \to h + fg$}
    \State{Evaluate the coefficient vector $C = eV^{-1}$. Elements of this vector are rational polynomials $C_i = x^{s_i}\frac{N_i(x)}{D_i(x)}$ where $N_i(x), D_i(x) \in GF(2)[x]$ and $N_i \equiv D_i \equiv 1 \mod x$}
    \For{$i{=}0$ to $2k{-}2$} 
        \State{Multiply $h$ by the adjoint of constant multiplication by $N_i(x)$} \Comment{Only $\cnotgate$s are used }
        \State{Multiply $h$ by the constant $D_i(x)$}\Comment{Only $\cnotgate$s are used }
        \State{Evaluate $f$ at the evaluation point. the result can be stored a register $\hat{f}$ of up to $\frac{m}{k} + t(k-1)$ qubits.}\Comment{Only $\cnotgate$s are used }
        \State{Repeat previous step for $g$ to get $\hat{g}$}\Comment{Only $\cnotgate$s are used }
        \State{Multiply $\hat{f}$ with $\hat{g}$ and store the result in $h[s_i:s_i + (2*|\hat{f}| - 1)]$} \Comment{This evaluates the point and multiplies by $x^{s_i}$ at the same time}
        \State{Undo steps 7-10}
    \EndFor

\State{Optimize the previous loop by reordering the evaluation points and grouping them based on common factors in their $N_i(x)$ and $D_i(x)$ to reduce $\cnotgate$ count}
\end{algorithmic}
\end{algorithm}
}

To construct circuits for Toom-$k$, we follow the procedure described in \alg{toom_cook_framework}. For example, to construct Toom-3, we use the evaluation points $\{(0, 1), (1, 1), (1, 0), (x, 1), (1, x)\}$ leading to the Vandermonde matrix $V_3$ and the coefficient vector $C_3$. The first, second, and fifth evaluations are simple to compute and are stored in $\frac{m}{3}$ qubits. The tricky cases are the third and fourth points since we need to use more than $\frac{m}{k}$ qubits. \alg{toom3_part3} gives the construction for the third evaluation point; the fourth point is similar.

Step 14 of \alg{toom_cook_framework} is to reorder and group terms to reduce the $\cnotgate$ count. To do this, notice that the coefficient for the point $(1, 0)$ is $x^n$ times the coefficient for $(0, 1)$, and the denominators for $(x, 1)$ and $(1, x)$ are the same, suggesting that we should group each pair together.  More optimizations are possible, and the use of local optimization will reduce the $\cnotgate$ count even further.

\begin{equation}
    V_3 = \begin{bmatrix}1 & 0 & 0 & 0 & 0\\1 & 1 & 1 & 1 & 1\\1 & x & x^{2} & x^{3} & x^{4}\\x^{4} & x^{3} & x^{2} & x & 1\\0 & 0 & 0 & 0 & 1\end{bmatrix}
\end{equation}

\begin{equation}
    C_3^T = \begin{bmatrix}
         \left(1 - x^{3 n}\right) + x^{2 n - 1} \left(x^{2} + x + 1\right) - x^{n-1} \left(x^{2} + x + 1\right)\\
        x^{n+1} \frac{x^{n-1} \left(x^{2} + 1\right) - 1 - x^{2 n}}{x^{2} - 2 x + 1}\\
        x^{n - 1} \frac{\left(- x^{n} - x^{n + 1} + x^{2 n + 1} + 1\right)}{x^{3} - x^{2} - x + 1}\\
        x^{n} \frac{\left(1 + x^{2 n-1} - x^{n-1} - x^{n}\right)}{x^{3} - x^{2} - x + 1}\\
        x^{n} \left(\left(x^{3 n} - 1\right) - x^{2 n - 1} \left(x^{2} + x + 1\right) + x^{n-1} \left(x^{2} + x + 1\right)\right)
    \end{bmatrix}
\end{equation}

{\centering
\begin{algorithm}[H]
\caption{The contribution of the evaluation point $(x, 1)$ in Toom-3} \label{alg:toom3_part3}
\begin{algorithmic}[1]	
    \State{Input: Quantum registers $f = f_0 + x^{m/3} f_1 + x^{2m/3} f_2$, $g = g_0 + x^{m/3} g_1 + x^{2m/3} g_2$, and $h$}
    \State{Output: a circuit applying $h \to h + (f_0 + x f_1 + x^2 f_2) (g_0 + x g_1 + x^2 g_2)$}
    \State{Multiply $h$ by the adjoint of constant multiplication by $x^{n} + x^{n + 1} + x^{2 n + 1} + 1$} 
    \State{Multiply $h$ by the constant $x^{3} + x^{2} + x + 1$}

    \State{$n = \lceil \frac{m}{3} \rceil$}
     \If{$m = 3 * n$}
        \State{$\hat{f}$ = f[:n] + f[-2:]}
        \State{$\hat{g}$ = g[:n] + g[-2:]}
        \State{$S = [0 \cdots n-1] \cup  \{m - 2, m - 1\}$}
    \Else
        \State{$\hat{f}$ = f[:n] + f[2 * n - 1 : 2 * n]}
        \State{$\hat{g}$ = g[:n] + g[2 * n - 1 : 2 * n]}
        \State{$S = [0 \cdots n-1] \cup  \{2n-1\}$}
    \EndIf

    \For{$i{=}0$ to $m-1$}
        \If{$i \notin S$}
            \State{$j = (i \mod n) + \lfloor \frac{i}{n} \rfloor$}
            \State{$\cnotgate(control{=}f[i], target{=}\hat{f}[j])$}
            \State{$\cnotgate(control{=}g[i], target{=}\hat{g}[j])$}
        \EndIf
    \EndFor
    \State{Multiply $\hat{f}$ with $\hat{g}$ and store the result in $h[n-1:n-1 + (2*|\hat{f}| - 1)]$}
    \State{Undo steps 4-22}
\end{algorithmic}
\end{algorithm}
}

\subsection{Results}

Table \ref{tab:toom3_cost} shows the cost of GF$(2)[x]$ multiplication using Toom-3 for powers of $3$. The number of $\toffoligate$ gates grows as $\mathcal{O}(m^{\log_3{5}})$, as expected.  However, because of the need to recurse on $\frac{m}{3}{+}2$ terms, we have a high leading constant yielding $\toffoligate$ count of $\approx 2.84 m^{\log_3{5}}$. This means that Toom-3 does not beat the Karatsuba algorithm on the $\toffoligate$ count until $m > 6{,}000$, and that is not mentioning the astronomical $\cnotgate$ count.

\begin{table}[]
    \centering
\begin{tabular}{rrr}
\toprule
m & Toffoli & $\cnotgate$ \\
\midrule
3 & 7 & 28 \\
9 & 67 & 1,380 \\
27 & 351 & 9,316 \\
81 & 1,783 & 53,668 \\
243 & 8,919 & 287,540 \\
729 & 44,527 & 1,492,460 \\
2,187 & 222,311 & 7,621,276 \\
6,561 & 1,110,483 & 38,575,428 \\
19,683 & 5,549,239 & 194,273,140 \\
59,049 & 27,737,287 & 975,541,380 \\
177,147 & 138,662,151 & 4,890,229,396 \\
531,441 & 693,245,623 & 24,488,748,508 \\
1,594,323 & 3,466,055,079 & 122,556,711,260 \\
4,782,969 & 17,329,818,307 & 613,123,014,500 \\
14,348,907 & 86,647,888,271 & 3,066,635,091,076 \\
43,046,721 & 433,236,280,503 & 15,336,240,137,508 \\
129,140,163 & 2,166,173,110,039 & 76,690,407,368,500 \\
387,420,489 & 10,830,843,812,287 & 383,479,690,809,420 \\
1,162,261,467 & 54,154,162,105,671 & 1,917,481,506,132,316 \\
\bottomrule
\end{tabular}
    \caption{Cost of GF$(2^m)$ multiplication using Toom-3.}
    \label{tab:toom3_cost}
\end{table}

\end{document}